\documentclass[sigconf]{acmart}

\copyrightyear{2024} \acmYear{2024} \setcopyright{acmlicensed}\acmConference[KDD '24]{Proceedings of the 30th ACM SIGKDD Conference on Knowledge Discovery and Data Mining}{August 25--29, 2024}{Barcelona, Spain} \acmBooktitle{Proceedings of the 30th ACM SIGKDD Conference on Knowledge Discovery and Data Mining (KDD '24), August 25--29, 2024, Barcelona, Spain} \acmDOI{10.1145/3637528.3671674} \acmISBN{979-8-4007-0490-1/24/08}
\AtBeginDocument{%
  \providecommand\BibTeX{{%
    \normalfont B\kern-0.5em{\scshape i\kern-0.25em b}\kern-0.8em\TeX}}}

\usepackage[scaled=.8]{beramono}
\usepackage{booktabs}
\usepackage{color}
\usepackage{colortbl}
\usepackage{tabularray}
\usepackage{float}
\AtBeginDocument{%
  \providecommand\BibTeX{{%
    \normalfont B\kern-0.5em{\scshape i\kern-0.25em b}\kern-0.8em\TeX}}}

\usepackage{makecell}
\usepackage{soul}
\usepackage{bbm}
\usepackage{amsthm}
\usepackage{balance}
\usepackage{multirow}
\usepackage[ruled,linesnumbered,vlined]{algorithm2e}
\usepackage{tcolorbox}
\usepackage[noend]{algpseudocode}
\usepackage[flushleft]{threeparttable}
\usepackage{fbox}
\usepackage{xcolor}
\usepackage{wasysym}
\usepackage{booktabs}
\usepackage{verbatim}
\usepackage{pgfplots}
\pgfplotsset{compat=1.18}
\usepackage{hyperref}
\usepackage{fontawesome}
\usepackage{tikz}
\usepackage{graphics}
\usepackage{multirow}
\usepackage{fancybox}
\usepackage{listings}
\usepackage{enumitem}
\usetikzlibrary{patterns}
\usepackage{amsmath, amsfonts}
\usepackage{graphicx}
\usepackage{subcaption}
\usepackage{textcomp}
\usepackage{xcolor}
\usepackage{cleveref}
\newtheorem{problem}{\textbf{Problem}}
\newtheorem{mydef}{Definition}[section]
\newtheorem{hyp}{Hypothesis}[section]
\newcommand{\defeq}{\overset{\text{\tiny def}}{=}}

\renewcommand\footnotetextcopyrightpermission[1]{}
\definecolor{add}{rgb}{0.2,0.4,0.6}

\acmConference[KDD 2024]{30th ACM SIGKDD Conference on Knowledge Discovery and Data Mining}{August 2024}{Barcelona, Spain}
\begin{document}

\title{Scalable Algorithm for Finding Balanced Subgraphs with Tolerance in Signed Networks}

\author{Jingbang Chen}
\authornotemark[1]
\email{j293chen@uwaterloo.ca}
\affiliation{%
  \institution{David R. Cheriton School of Computer Science, University of Waterloo}
  \city{Waterloo}
\country{Canada}
}

\author{Qiuyang Mang}
\authornotemark[1]
\email{qiuyangmang@link.cuhk.edu.cn}
\affiliation{%
  \institution{School of Data Science, The Chinese University of Hong Kong, Shenzhen}
  \city{Shenzhen}
\country{China}
}

\author{Hangrui Zhou}
\authornote{The first three authors contributed equally to this research.}
\email{zhouhr23@mails.tsinghua.edu.cn}
\affiliation{%
  \institution{Institute for Interdisciplinary Information Sciences (IIIS), Tsinghua University}
\city{Beijing}
\country{China}
}

\author{Richard Peng}
\email{yangp@cs.cmu.edu}
\affiliation{%
  \institution{Computer Science Department, Carnegie Mellon University}
  \city{Pittsburgh}
\country{USA}
}

\author{Yu Gao}
\email{ygao2606@gmail.com}
\affiliation{%
  \institution{Independent}
  \city{Beijing}
  \country{China}
\country{}
}

\author{Chenhao Ma}
\authornote{Chenhao Ma is the corresponding author.}
\email{machenhao@cuhk.edu.cn}
\affiliation{%
  \institution{School of Data Science, The Chinese University of Hong Kong, Shenzhen}
  \city{Shenzhen}
\country{China}
}



\newcommand{\todo}[1]{\textcolor{red}{[todo: #1]}}
\newcommand{\jb}[1]{\textcolor{blue}{[cjb: #1]}}
\newcommand{\yu}[1]{\textcolor{green}{[gy: #1]}}
\newcommand{\mqy}[1]{\textcolor{violet}{[mqy: #1]}}
\newcommand{\hhz}[1]{\textcolor{teal}{[hhz: #1]}}
\newcommand{\mch}[1]{\textcolor{purple}{[mch: #1]}}
\newcommand{\VI}{\texttt{VertexInsert}}
\newcommand{\VD}{\texttt{VertexDelete}}
\newcommand{\CF}{\texttt{ColorFlip}}
\newcommand{\pr}[2]{\mathsf{Pr}_{#1}\left[#2\right]}
\newcommand{\expec}[2]{\mathbb{E}_{#1}\left[#2\right]}

\begin{abstract}
 Signed networks, characterized by edges labeled as either positive or negative, offer nuanced insights into interaction dynamics beyond the capabilities of unsigned graphs. Central to this is the task of identifying the maximum balanced subgraph, crucial for applications like polarized community detection in social networks and portfolio analysis in finance. Traditional models, however, are limited by an assumption of perfect partitioning, which fails to mirror the complexities of real-world data. Addressing this gap, we introduce an innovative generalized balanced subgraph model that incorporates tolerance for imbalance. Our proposed region-based heuristic algorithm, tailored for this \textbf{NP}-hard problem, strikes a balance between low time complexity and high-quality outcomes. Comparative experiments validate its superior performance against leading solutions, delivering enhanced effectiveness (notably larger subgraph sizes) and efficiency (achieving up to 100$\times$ speedup) in both traditional and generalized contexts.
\end{abstract}

\maketitle

\section{INTRODUCTION}


Social media platforms, integral to our digital connectivity, transform interactions into analyzable social networks. By deploying graph algorithms, we discern network properties like community detection~\cite{10.1145/3589314,fortunato2010community} and partitioning~\cite{bulucc2016recent}, informing user experience improvements and recommendation systems~\cite{wu2022graph}. Yet, these platforms can also engender echo chambers that reinforce divisive ideologies, challenging democratic health. Consequently, detecting and countering polarization in social networks is a critical area of research~\cite{baumann2020modeling,nguyen2014exploring}, pivotal for developing defenses against misinformation~\cite{chitra2020analyzing,banerjee2023mitigating}.

A classical model that applies to social networks to deal with polarization is the signed graphs. The signed graph model overcomes the limitation that normal graphs cannot capture users' dispositions. Generally speaking, while the vertex set represents users, there are two kinds of edges between vertices indicating agreement or disagreement. We often refer to them as the positive and the negative edges. The signed graph model was first introduced by Harary in 1953~\cite{harary1953notion} to study the concept of \textit{balance}. The concept of \textit{balance} is important in signed graphs. Generally speaking, a signed graph is balanced if it can be decomposed into two disjoint sets such that positive edges are between vertices in the same set while negative edges are between vertices from different sets. Such a concept has many practical applications, especially in the polarization study. In a social network, a balanced graph suggests two communities exist with contrasting relationships while maintaining inner cohesion. 

\begin{figure}[t!]
\vspace{-1mm}
  \centering
  \includegraphics[width=1\linewidth]{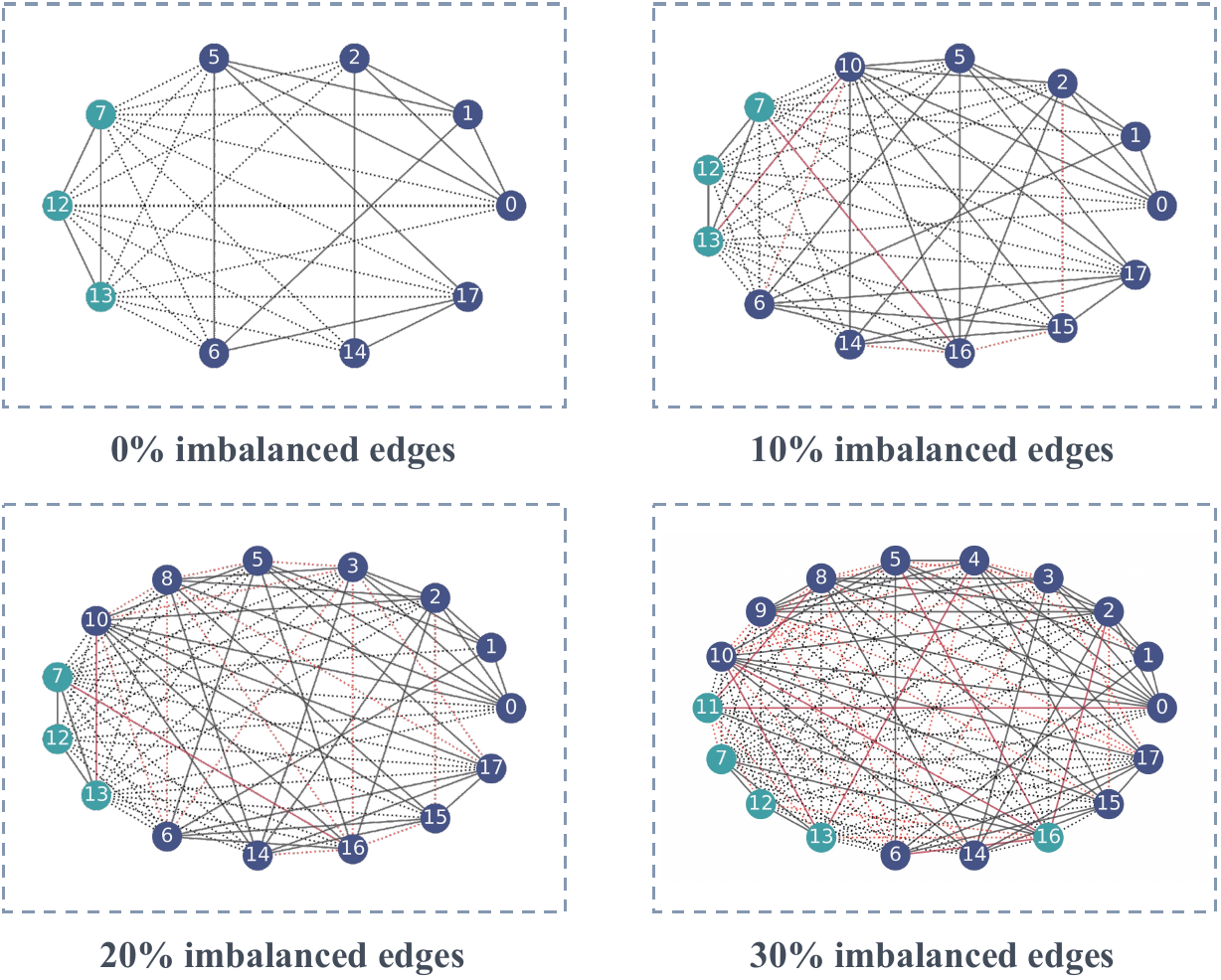}
  \caption{Balanced subgraphs are found in \textsc{Cloister} with different tolerance, where solid edges are positive, dashed edges are negative, black edges are balanced, and red edges are imbalanced.}
  \label{fig:intro}
\end{figure}

There are two lines of work when studying social network polarization with signed graphs. Since most graphs coming from real scenarios are not balanced, if we can find the maximum balanced subgraph (MBS) instead, it usually reveals the largest polarized communities along with several important properties. 

On the other line of work, instead of extracting the maximal subgraph, it focuses on removing edges to guarantee the balance of the original graph. The minimum number of edges whose deletion makes all connected components balanced is called the \textit{Frustration Index} of the given signed graph.

However, such a notion of balance is strict in that no edge can disobey the condition. In reality, such strictness does not usually appear. 
For example, though being dominated, there usually exists some voices of disagreement on the majority idea, even in the most extreme community. Besides, two individuals in different political parties might reach a consensus on certain issues. In other words, there usually exist some \textit{imbalanced} edges in signed graphs extracted from the community, which are against the strict notion of balance. Failing to handle these imbalanced edges might prevent us from identifying, extracting, and characterizing the community properly in real scenarios. We provide an example in Figure~\ref{fig:intro}\footnote{konect.cc/networks/moreno\_sampson}: As we increase the limitation for imbalanced edges, the community is getting significantly larger and denser. Therefore, we may be unable to capture the actual community if we prohibit imbalanced edges when computing.

Many questions arise here due to the many limitations: 
\begin{itemize}[leftmargin=*]
    \item Can these two lines of work be unified? 
    \item Can we develop algorithms finding communities that are non-strictly balanced? 
\end{itemize}

In this paper, we answer these questions affirmatively. We design a tailored function to measure the tolerance of the imbalance. To properly describe such tolerance, we adapt the frustration index as part of the tolerance function. In this way, we can either allow no tolerance to extract the MBS in the network or a loose tolerance to extract communities that might contain voices of disagreement. We further raise new problems based on this tolerance model and present a new region-based heuristic algorithm that computes maximal balanced subgraphs under the tolerance setting. 
Our new algorithm is versatile, delivering high-quality solutions across a range of problems based on the tolerance model, and it's also efficient, with an expected output time proportional to the size of the result.
By setting different tolerances, we utilize our algorithm to handle different tasks and compare it with the performance of state-of-the-art algorithms\footnote{Here, the SOTA algorithms are adapted to fit our new tasks under various tolerances.} on 8 real-world datasets.

Our contributions are summarized as follows:
\begin{itemize}[leftmargin=*]
    \item We introduce a novel, generalized, and practical model for identifying balanced subgraphs with tolerance in signed graphs.
    \item We have developed an efficient region-based heuristic randomized algorithm, characterized by an expected time complexity proportional to the size of the output, and coupled with a guarantee of result quality.
    \item Extensive experiments show that our algorithm consistently outperforms the baselines in terms of the quality of the returned subgraphs and achieves up to 100$\times$ speedup in terms of running time.
    \item The effectiveness of our algorithm is also evident in its application to polarized community detection, as detailed in Section~\ref{sec:2pc}.
\end{itemize}

{\bf Outline.} The rest of the paper is organized as follows. We review the related work in \Cref{sec:related}, and introduce our generalized maximum balanced subgraph model with tolerance in \Cref{sec:preliminaries}. \Cref{sec:algo} presents our region-based heuristic algorithm. Experimental results are given in \Cref{sec:eval}, and we conclude in \Cref{sec:conc}.

\section{RELATED WORK}
\label{sec:related}
\paragraph{Signed Graphs}
Signed graphs were first introduced by Harary in 1956 to study the notion of \textit{balance}~\cite{harary1953notion}. Cartwright and Harary generalized Heider's theory of balance onto signed graphs~\cite{cartwright1956structural}. Harary also developed an algorithm to detect balance in signed graphs~\cite{harary1980simple}. There are other works on studying the minimum number of sign changes to make the graph balance~\cite{akiyama1981balancing}. Spectral properties have also been studied recently. Hou et al. have studied the smallest eigenvalue in signed graphs' Laplacian~\cite{hou2003laplacian,hou2005bounds}.

Many works focus on community detection or partition in signed graphs. 
Anchuri et al. give a spectral method that partitions the signed graph into non-overlapping balanced communities~\cite{anchuri2012communities}. Doreian and Mrvar propose an algorithm for partitioning a directed signed graph and minimizing a measure of \textit{imbalance}~\cite{doreian1996partitioning}. Yang et al. give a random-walk-based method to partition into cohesively polarized communities~\cite{yang2007community}. The more recent work is by Niu et al~\cite{CohesivelyPolarized}. They leverage the balanced triangles, which model cohesion and polarization simultaneously, to design a good heuristic.

\paragraph{Maximum Balanced Subgraphs}
The Maximum Balanced Subgraphs (MBS) problem has two variants: maximizing the number of vertices (MBS-V) and edges (MBS-E). Poljak et al. give a tight lower bound in 1986~\cite{poljak1986polynomial} on the number of edges and vertices of the balanced subgraph. 
The MBS-E problem is in fact \textbf{NP}-hard since it can be formulated as a generalization of the standard \textsc{MaxCut} problem. To find the exact solution, there are some algorithms in the context of \textit{fixed-parameter tractability} (FPT) are developed~\cite{huffner2007optimal,crowston2013maximum}. More studies are on extracting large balanced subgraphs. DasGupta et al. have developed an algorithm based on semidefinite programming relaxation (SDP)~\cite{dasgupta2007algorithmic}. For the MBS-V problem, Figueiredo and Frota propose several heuristic methods in 2014~\cite{figueiredo2014maximum}.

In 2020, Ordozgoiti et al. proposed a new algorithm named \textsc{Timbal}~\cite{ordozgoiti2020finding} that extracts large balanced subgraphs regarding both vertices and edges. \textsc{Timbal} relies on signed spectral theory and the bound for perturbations of the graph Laplacian. However, Timbal is not stable and the balanced subgraph it found is sometimes small and unsatisfying as shown in our experiments. 

\paragraph{Frustration Index} The frustration index was first introduced in 1950s \cite{abelson1958symbolic,harary1959measurement}. Computing the frustration index is related to the \textsc{EdgeBipartization} problem, which requires minimization of the number of edges whose deletion makes the graph bipartite. Since \textsc{EdgeBipartization} is \textbf{NP}-hard, computing the frustration index is also \textbf{NP}-hard. The \textsc{MaxCut} is also a special case of the frustration index problem. Assuming Khot's unique games conjecture~\cite{khot2002power}, it is still \textbf{NP}-hard to approximate within any constant factor. For the non-constant factor case, there are works that produce a solution approximated to a factor of $O(\sqrt{\log n})$~\cite{agarwal2005log} or $O(\log k)$~\cite{avidor2007multi} where $n$ is the number of vertices and $k$ is the frustration index. Coleman et al. have given a review on different approximation algorithms~\cite{coleman2008local}. Hüffner et al. show that the frustration index is \textit{fixed parameter tractable} and can be computed in $O(2^k m^2)$, where $m$ is the number of edges and $k$ is the fixed parameter (the frustration index)~\cite{huffner2010separator}. There are also algorithms using binary programming models to compute the exact frustration index~\cite{aref2018computing,aref2020modeling}. 

\section{Problem Specification}
\label{sec:preliminaries}


A signed graph is an undirected simple graph $G = (V, E^{+}, E^{-})$ where $V$ is the vertex set and $E^{+}, E^{-}$ represent the positive and negative signed edge sets. We first give the formal definition of balanced graphs as follows, which is the same as the previous works~\cite{ordozgoiti2020finding, CohesivelyPolarized}.
Note that we and these works require the graph to be connected for the community detection proposal. 

\begin{mydef}[Balanced Graph]
\label{def:balanced}
    Given a signed graph $G = (V, E^{+}, E^{-})$, $G$ is balanced if $G$ is connected and there exists a partition $V = V_1 \cup V_2, V_1 \cap V_2 = \emptyset$ such that for each edge $(i, j) \in E^{+}$, vertices $i$ and $j$ belong to the same set within $V_1$ and $V_2$, while for each edge $(i, j) \in E^{-}$, they belong to the different sets. 
\end{mydef}

Graphs are usually not balanced, especially when they are from practical scenarios. 
Therefore, people turn to study to find the maximum balanced subgraph (MBS) from the given graph. There are usually two variants of problems: maximizing the number of vertices~\cite{ordozgoiti2020finding} and edges~\cite{dasgupta2007algorithmic}. 

\begin{problem}[MBS-V]
    \label{p1}
    Given a signed graph $G = (V, E^{+}, E^{-})$, find the graph $G' = (V', E'^{+}, E'^{-})$ induced by $V' \subseteq V$ such that $G'$ is balanced and $|V'|$ is maximized.
\end{problem}
\begin{problem}[MBS-E]
    \label{p2}
    Given a signed graph $G = (V, E^{+}, E^{-})$, find the graph $G' = (V', E'^{+}, E'^{-})$ induced by $V' \subseteq V$ such that $G'$ is balanced and $|E'^{+} \cup E'^{-}|$ is maximized.
\end{problem}

\begin{figure}[t!]
\vspace{-1mm}
  \centering
  \includegraphics[width=0.9\linewidth]{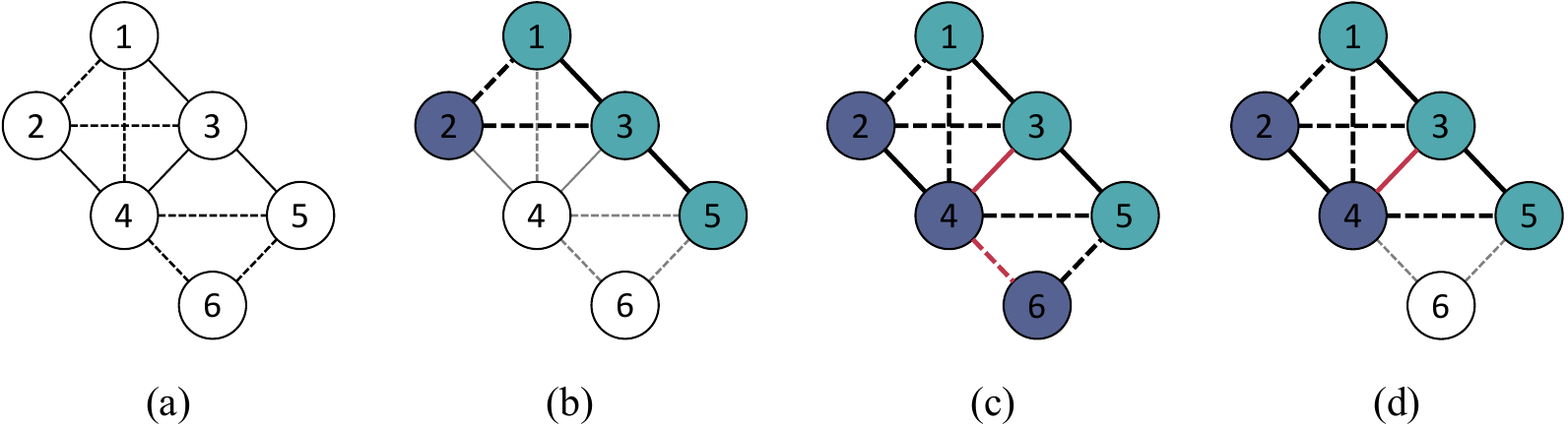}
  \caption{A signed graph (a), its MBS-V and MBS-E (b), its TMBS-V and TMBS-E (c), and its $\beta$-TMBS (d) with $\beta=\frac{1}{3}$.}
  \label{fig:toygraphs}
\end{figure}

We illustrate the concepts of MBS-V and MBS-E with an example in Figure~\ref{fig:toygraphs}, along with other problems to be discussed below. Here, we give a graph with 6 vertices numbered from $1$ to $6$ where solid edges are positive and dashed edges are negative, presented in (a). When solving \textsc{Problem}~\ref{p1} or ~\ref{p2} on this graph, the optimal subgraph is constructed by vertex $1,2,3,5$, where $V_1=\{1,3,5\}$ and $V_2=\{2\}$. They are painted in two different colors in (b). 

Another well-known problem in this topic is computing the minimum number of edges whose deletion makes all connected components balanced. This minimum amount of the edge removal is called the \textit{Frustration Index} of the given signed graph $G$, denoted as $L(G)$.

Now, we are ready to introduce the concept of \textit{Balance with $\beta$-tolerance} in signed networks, which is central to our paper.
\begin{mydef}[Balanced Graph under $\beta$-Tolerance]
   Given a signed graph $G = (V, E^{+}, E^{-})$ and a tolerance parameter $\beta \in (0, 1]$, $G$ is balanced under $\beta$-Tolerance if $G$ is connected and $\frac{L(G)}{|E^{+} \cup E^{-}|} \leq \beta$.
\end{mydef}

In other words, when a graph is balanced under $\beta$-tolerance, it implies that by removing at most $\beta |E^{+} \cup E^{-}|$ edges from it, we can ensure that all connected components of the remaining graph become strictly balanced.
Expanding upon the original definition of balance, this concept allows for the tolerance of a maximum of $\beta |E^{+} \cup E^{-}|$ edges that disobey the partitioning of polarized communities.

Under such $\beta$-tolerance restriction, finding the maximum balanced subgraph is generalized into the following two problems.

\begin{problem}[TMBS-V]
    \label{p3}
    Given a signed graph $G = (V, E^{+}, E^{-})$ and a tolerance parameter $\beta \in (0, 1]$, find the graph $G' = (V', E'^{+}, E'^{-})$ induced by $V' \subseteq V$ such that $G'$ is balanced under $\beta$-Tolerance and $|V'|$ is maximized.
\end{problem}
\begin{problem}[TMBS-E]
    \label{p4}
    Given a signed graph $G = (V, E^{+}, E^{-})$ and a tolerance parameter $\beta \in (0, 1]$, find the graph $G' = (V', E'^{+}, E'^{-})$ induced by $V' \subseteq V$ such that $G'$ is balanced under $\beta$-Tolerance and $|E'^{+} \cup E'^{-}|$ is maximized.
\end{problem}


When $\beta$ becomes any value between $(0, \frac{1}{|E^{+}\cup E^{-}|})$, \textsc{Problem}~\ref{p3} and \textsc{Problem}~\ref{p4} are equivalent to the \textsc{Problem}~\ref{p1} and \textsc{Problem}~\ref{p2} respectively. 
However, unlike the strictly balanced graph, even deciding whether a signed graph satisfies $\beta$-tolerance is hard. Formally, we have the following lemma.
The proof is deferred to Appendix~\ref{appendix:sec3}.

\begin{lemma}
\label{lemma:hardness}
Given a signed graph $G = (V, E^{+}, E^{-})$, deciding whether $G$ is balanced under $\beta$-tolerance cannot be done within polynomial time for tolerance parameter $\beta > \frac{1}{|E^{+}\cup E^{-}|}$, unless \textbf{P $=$ NP}. 
\end{lemma}

Therefore, it is difficult to find large balanced graphs with tolerance and we cannot adapt any previous algorithms on MBS or other related problems. Meanwhile, there is another limitation that is worth mentioning. If $\beta \geq \frac{L(G)}{|E^{+} \cup E^{-}|}$, \textsc{Problem}~\ref{p3} and \textsc{Problem}~\ref{p4} has a trivial optimal solution: the largest connected component. That is to say, in such loose tolerance restriction, the optimal solution to these two problems might not reflect the polarized community. The previous showcase given in Figure~\ref{fig:toygraphs} (c) is a typical example: When we solve \textsc{Problem}~\ref{p3} and \ref{p4} with $\beta=\frac{1}{3}$ in such graph, we get the whole graph. However, vertex $6$ is in fact a noise that only has negative edges connected to both $V_1$ and $V_2$. Therefore, it should not be included in any meaningful subgraph solution. 

To encounter the challenges and limitations, we first propose \textsc{Tolerant Balance Index}, a novel metric to evaluate signed graphs under the given tolerance for the balance. 

\begin{mydef}[Tolerant Balance Index (TBI)]
    Given a signed graph $G = (V, E^{+}, E^{-})$ and a tolerance parameter $\beta \in (0, 1]$, define its Tolerant Balance Index $\Phi(G, \beta)$ as the value of $|E^{+}\cup E^{-}| - \frac{L(G)}{\beta}$. 
\end{mydef}

Correspondingly, we propose the following problem as the main task to solve throughout the paper.

\begin{problem}[$\beta$-TMBS]
    \label{p5}
    Given a signed graph $G$ and a tolerance parameter $\beta \in (0, 1]$, find the graph $G'$ induced by $V' \subseteq V$ such that $G'$ is connected and $\Phi(G', \beta)$ is maximized.
\end{problem}

It is easy to show, for any balanced graph $G$ under $\beta$-tolerance, $\Phi(G, \beta) \geq 0$. To discuss the usage of solving \textsc{Problem}~\ref{p5}, we also want to address that maximizing the number of edges (\textsc{Problem}~\ref{p2} and~\ref{p4}) approximates the solution for maximizing the number of vertices (\textsc{Problem}~\ref{p1} and~\ref{p3}). That is to say, these two variants of MBS problems are related. 

\begin{lemma}
\label{lemma:samequestion}
    Given a signed graph $G$, the solutions to \textsc{Problem}~\ref{p2} and \textsc{Problem}~\ref{p4} are $\frac{1}{\Delta}$-approximations for the \textsc{Problem}~\ref{p1} and \textsc{Problem}~\ref{p3} respectively, where $\Delta$ is the maximum degree of vertices in $G$. 
\end{lemma}

We defer the proof of Lemma~\ref{lemma:samequestion} to Appendix~\ref{appendix:sec3}. Considering the aforementioned interrelationship between two variants, we propose that finding the large connected subgraph with maximal $\Phi(G, \beta)$ (\textsc{Problem}~\ref{p5}) can be considered as a good approximation of the solutions to \textsc{Problem}~\ref{p3} and \textsc{Problem}~\ref{p4} simultaneously. In addition, since \textsc{Problem}~\ref{p5} takes both the size and the polarity into account, it mitigates the limitation of large $\beta$. When we solve \textsc{Problem}~\ref{p5} with $\beta=\frac{1}{3}$ on Figure~\ref{fig:toygraphs} (a), the optimal subgraph will be vertices excluded vertex $6$, where $V_1=\{1,3,5\}$ and $V_2=\{2,4\}$, presented in (d). As it is shown, the noise vertex $6$ is not mistakenly selected into the optimal subgraph we try to find. Therefore, the polarity is being preserved. 

In Section~\ref{sec:algo}, we will present an efficient and effective algorithm for \textsc{Problem}~\ref{p5}. Then in Section~\ref{sec:eval}, we will demonstrate the experimental results to support our proposal.

\SetAlCapNameFnt{\small}
\SetAlCapFnt{\small}

\section{ALGORITHM}
\label{sec:algo}
We propose \textit{\underline{R}egion-based \underline{H}euristic (\textsc{RH})}, a new algorithm that searches for large balanced subgraphs with the $\beta$-tolerance restriction (\textsc{Problem}~\ref{p5}). 
Our algorithm runs in a given signed graph $G = (V, E^{+}, E^{-})$, where $E^{+}, E^{-}$ represent the positive and negative signed edge sets. The tolerance parameter $\beta$ is also given as input.

\subsection{Relaxation}\label{sec:relaxation}
It is not easy to compute $\Phi(G,\beta)$ directly. Instead, we propose an alternative method to approximate it from the lower end while guaranteeing the tolerance condition ($\Phi(G,\beta)\ge 0$) is not violated. 

For each vertex in $G$, we assign a color in $\{0,1\}$. Only vertices of the same color can be grouped in the same community. We denote any color assignment of $G$ as $\mathcal{X}$. Such definition is aligned with the previous work~\cite{aref2020modeling}. Now, we are ready to give the formal definition of our newly proposed \textit{Tolerant Balanced Count}.

\begin{mydef}[Tolerant Balance Count (TBC)]
Given a signed graph $G = (V, E^{+}, E^{-})$ under coloring $\mathcal{X}$ and a tolerance parameter $\beta \in (0, 1]$, define its Tolerant Balance Count as
\begin{equation*}
\small
\begin{aligned}
    \hat{\Phi}(G, \beta, \mathcal{X}) = |E^{+} \cup E^{-}| - \frac{1}{\beta}\sum_{(i, j) \in E^{+}} \mathbbm{1}\{x_i \neq x_j\} 
    -\frac{1}{\beta} \sum_{(i, j) \in E^{-}} \mathbbm{1}\{x_i = x_j\}
\end{aligned}
\end{equation*}
\end{mydef}

The following lemma states that $\Phi(G,\beta)$ is, in fact, the upper bound of $\hat{\Phi}(G,\beta,\mathcal{X})$ with respect to different coloring $\mathcal{X}$. The proof is deferred to Appendix~\ref{appendix:sec4}.
We also have the following corollary that ensures the $\beta$-tolerance requirement is not violated. 

\begin{lemma}
\label{lemma:upperbound}
Given a signed graph $G = (V, E^{+}, E^{-})$ and a tolerance parameter $\beta \in (0, 1]$, we have $\Phi(G, \beta) = \max_{\mathcal{X}} \hat{\Phi}(G, \beta, \mathcal{X})$.
\end{lemma}

\begin{corollary}\label{corol:count}
Given a signed graph $G$ and a tolerance parameter $\beta$, $G$ is balanced under $\beta$-tolerance if there exists a coloring $\mathcal{X}$ such that $\hat{\Phi}(G, \beta, \mathcal{X}) \geq 0$.
\end{corollary}

In this way, although we cannot compute $\Phi(G,\beta)$ directly, we can approximate the actual value by accumulating the non-negative $\hat{\Phi}(G, \beta, \mathcal{X})$ values of various coloring. Additionally, by Corollary~\ref{corol:count}, we can guarantee the $\beta$-tolerance during the whole computation. Such relaxation plays an important role in our algorithm, and the experimental results in Section~\ref{sec:eval} also validate its effectiveness.

\subsection{Local Search}\label{sec:localsearch}
Our search process in Algorithm~\ref{alg:hillclimbing} starts from an initial vertex $s$ instead of the whole vertices set. We will discuss how to choose $s$ wisely in the later Section~\ref{sec:sampling}.

\paragraph{Search Operations}As discussed in Section~\ref{sec:relaxation}, we will search for a connected subgraph $G'$ and a coloring $\mathcal{X}$ such that $\hat{\Phi}(G', \beta, \mathcal{X})$ is as large as possible. In the following, if not specified, we use $\hat{\Phi}$ to denote the tolerant balanced count of the subgraph. Throughout the search process, we use a set $S$ to store the selected vertices and the coloring simultaneously. Specifically, $S$ stores selected tuples $(x,c)$, where $(x,c)$ represents a vertex $x$ colored as $c$. There are three basic operations with $S$ that will insert, delete, or change the color of a vertex $x$ respectively. Note that these three operations well define the neighborhood of any solution we find. 

\begin{itemize}[leftmargin=*]
    \item $\VI(x,c)$: Execute $S \leftarrow S\cup(x,c)$.
    \item $\VD(x)$: Let $c$ be the color of $x$. Execute $S \leftarrow S\ \backslash\ (x,c)$. 
    \item $\CF(x)$: Let $c$ be the color of $x$ and $\overline{c}$ be $c$'s opposite color. Execute $\VD(x)$ and $\VI(x,\overline{c})$ in order.
\end{itemize}

Since our search begins from $s$, we initialize $S$ to be $\{(s,0)\}$ (line~\ref{search:init}). We use $\hat{\Phi}_{opt}$ and $\hat{\Phi}_{cur}$ to store the maximum $\hat{\Phi}$ that we have found and the current $\hat{\Phi}$ respectively. 
Initially, they are initialized to be $0$ (line~\ref{search:init}). 
During our process, if we choose to execute $\VI(x,c)$ or $\CF(x)$ for some vertex $x$, we will greedily execute one with the maximal increment of $\hat{\Phi}$ that it can contribute. Therefore, we use two max-heaps $MH_{I},MH_{F}$ to assist. 
After being initialized (line~\ref{search:initheap}), since $s$ is selected, for each neighbor $v$ of $s$, we calculate the corresponding contribution to $\hat{\Phi}$ when $\VI(v,0)$ or $\VI(v,1)$ are executed and insert them to $MH_{I}$ (line~\ref{search:mhiinit1} to ~\ref{search:mhiinit2}). 
For $MH_F$, since only $s$ is selected, we calculate the contribution for $\CF(x)$ and insert it into $MH_F$ (line~\ref{search:mhfinit}). 

Here, we do not use any structure to store the contribution when deleting any vertex from $S$. Instead, whenever we want to find the optimal deletion, we can compute for every vertex in $S$ altogether with a total cost of $O(n)$, for which we may need to use the well-known Tarjan algorithm~\cite{tarjan} to guarantee the subgraph stays connected after the deletion. 
We denote such process as $\texttt{DelEval}(S)$.

Our search method is to execute one of the three operations repeatedly. If we only consider inserting vertices from $S$'s neighbors, the total number of $\VI$ is bounded by the size of the graph. However, since we also have the other two operations, which are \textit{non-incremental}, if we do not design a proper termination strategy, the time complexity might become exponential. 
Specifically, we have two parameters that help to define the termination strategy of the algorithm: A float number $p \in (0,1)$ denotes the \textit{non-incremental probability} and an integer $T$ to limit the number of potentially wasted operations. 

\paragraph{Operation Selection} We first discuss how to select an operation each time. Here, we will use the non-incremental probability $p$ to help determine. 
We use $op$ to denote the operation we select and $\Delta_{\hat{\Phi}}$ to denote the corresponding contribution value to $\hat{\Phi}$. 
These two variables are initialized by acquiring the best $\VI$ operation from $MH_I$ (line~\ref{search:getvi}), indicating choosing an insertion. 
Then, we generate a random float number $z$ from uniform distribution $U(0,1)$. 
If $z<p$, we also consider using a $\CF$ operation: We acquire the best $\CF$ operation from $MH_F$ and update the two variables if the corresponding $\Delta_{\hat{\Phi}}$ is larger than the current one (line~\ref{search:getcf1} to ~\ref{search:getcf2}).
Similarly, we regenerate $z$ and if $z< \frac{p\log |S|}{|S|}$, we try choosing a $\VD$ operation: After calling $\texttt{DelEval}(S)$ to recalculate for every possible vertex deletion, we choose the best among them and try updating (line~\ref{search:getvd1} to ~\ref{search:getvd3}). 
In this way, we choose an $\VI$ operation by default, and with some probability, we check if the current optimal $\CF$ and $\VD$ can contribute more. The two probabilities are by design and help to balance between accuracy and efficiency.


We execute the current $op$ operation after the selection phase and also update the current $\hat{\Phi}_{cur}$. Since $S$ is updated after the execution, we must also update $MH_I$ and $MH_F$ correspondingly. This is done in a similar way as the initialization (line~\ref{search:updates1} to ~\ref{search:updates5}).

\paragraph{Search Termination}
When the current selected operation enlarges the current $\hat{\Phi}_{cur}$, we denote such operation as a \textit{progressive} one. Otherwise, it is \textit{non-progressive}. We use a parameter $T$ to prevent too many \textit{non-progressive} operations. Specifically, our search will terminate when the number of \textit{non-progressive} operations exceeds $T$ times the number of \textit{progressive} operations. 
We implement such a strategy with a counter $t$. After we select and execute an operation (line~\ref{search:getvi} to ~\ref{search:updates5}), we compare the current $\hat{\Phi}_{cur}$ with $\hat{\Phi}_{opt}$. If $\hat{\Phi}_{cur}$ is smaller, we decrease $t$ by $1$. Otherwise, we update $\hat{\Phi}_{opt}$ and increase $t$ by $T$ (line~\ref{search:updatet1} to ~\ref{search:updatet2}). In this way, if $t<0$, the search should terminate (line~\ref{search:loop}). In addition, we also terminate the search when $S$ contains all vertices in $G$ since no further insertion can be executed (line~\ref{search:loop}).
After the search, we undo the last few non-progressive operations to retract the optimal subgraph we have found (line~\ref{search:undo}). We return with this subgraph as $\overline{H}$ and its corresponding $\hat{\Phi}$ (line~\ref{search:return}).

\begin{algorithm}[t]
\caption{\textsc{RH: Search}}\label{alg:hillclimbing}
\KwIn{Signed graph $G$; Tolerance parameter $\beta$; Non-incremental probability $p$; Initial vertex $s$; Early stop turn limit $T$.}
\KwOut{$\overline{H} \subseteq G$: the found balanced subgraph with $\beta$-tolerance; $\hat{\Phi}$: the lower bound of $\overline{H}$'s balance value we achieved.}
$S \leftarrow \{(s, 0)\}$, $\hat{\Phi}_{opt} \leftarrow 0$, $\hat{\Phi}_{cur} \leftarrow 0$, $t \leftarrow T$; \\\label{search:init}
Initialize max-heaps $MH_{I},  MH_{F}$; \\ \label{search:initheap}
Insert $s$ into $MH_F$; \\  \label{search:mhfinit}
\For{$v \in N(s)$}{ \label{search:mhiinit1}
Insert $(v, 0)$ and $(v, 1)$ into $MH_I$; \label{search:mhiinit2}
}

\While {$t \geq 0 \text{ and } |S| < |V(G)|$}{ \label{search:loop}
    $\{op, \Delta_{\hat{\Phi}}\} = \texttt{GetTopNode}(MH_I)$\tcp*{ordered by $\Delta_{\hat{\Phi}}$} \label{search:getvi}
    \If{$z \sim U(0, 1) < p$}{ \label{search:getcf1}
        $\{op, \Delta_{\hat{\Phi}}\} \leftarrow \max(\{op, \Delta_{\hat{\Phi}}\}, \texttt{GetTopNode}(MH_F)$); \label{search:getcf2}
    }
    \If{$z \sim U(0, 1) < \frac{p\log |S|}{|S|}$}{ \label{search:getvd1}
        $\{\hat{op}, \overline{\Delta_{\hat{\Phi}}}\} \leftarrow$ \texttt{DelEval}$(S)$; \\ \label{search:getvd2}
        $\{op, \Delta_{\hat{\Phi}}\} \leftarrow \max(\{op, \Delta_{\hat{\Phi}}\}, \{\hat{op}, \overline{\Delta_{\hat{\Phi}}}\})$; \label{search:getvd3}
    }
    $S \leftarrow $ Execute operation $op$ on $S$; \\ \label{search:updates1}
    $x \leftarrow$ the vertex of the $op$; \\ \label{search:updates2}
    \For{$v \in N(x) \cup x$}{ \label{search:updates3}
        Insert, update or delete $v$'s operations in $MH_{I}$ and $MH_{F}$; \\ \label{search:updates4}
    }
    $\hat{\Phi}_{cur} \leftarrow \hat{\Phi}_{cur} + \Delta_{\hat{\Phi}}$; \\ \label{search:updates5}
    \lIf{$\hat{\Phi}_{cur} \leq \hat{\Phi}_{opt}$}{$t \leftarrow t - 1$} \label{search:updatet1}
    \lElse{
        $\hat{\Phi}_{opt} \leftarrow \hat{\Phi}_{cur}$, $t \leftarrow t + T$ \label{search:updatet2}
    }

}
Undo the last non-progressive operations on $S$; \\ \label{search:undo}
\Return $\overline{H} = G[S], \hat{\Phi} = \hat{\Phi}_{opt}$; \label{search:return}
\end{algorithm}

\subsection{Region-based Sampling}
\label{sec:sampling}
In the previous section, we propose a search process that starts from an arbitrary vertex $s$. It is reasonable to foresee that the choice of $s$ might affect the result significantly. If we start only from too few vertices, our result in the end might be some local maximal solutions, which would be much worse than the global maximal one. One of the solutions is to enumerate all possible $s$, \textit{i.e.}, all vertices in $G$. However, such pure enumeration may result in excessively high time complexity. To balance between performance and efficiency, we propose a \textit{Region-based Sampling} strategy.

Our sampling method is mainly based on two hypotheses. The first hypothesis indicates that the probability of finding a nearly optimal subgraph is high if we are able to select some vertices in the optimal subgraph as the starting vertex.  Here, the `optimal' denotes the solution we found by enumerating all vertices as the starting vertex. We formally state such a hypothesis as follows.

\begin{hyp}
\label{h1}
Given a signed graph $G$ and a tolerance parameter $\beta$, suppose the optimal subgraph found by Algorithm~\ref{alg:hillclimbing} starting with vertex $x$ is $G_x$, and the optimal graph among all $G_x$ is $G_{opt}$.
For the given $\epsilon$, there exists a subset $V' \subseteq V(G_{opt})$ with $\frac{|V'|}{|V(G_{opt})|} \geq \frac{1}{2}$ such that
$\Phi(G_{x}, \beta) \geq (1-\epsilon) \Phi(G_{opt}, \beta), \forall x \in V'$.
\end{hyp}

Another hypothesis describes the relation between $\hat{\Phi}$ values returned by two different calls of Algorithm~\ref{alg:hillclimbing}, if we select different starting vertex. We argue that if the return $\hat{\Phi}$ value is larger, the found subgraph will likely be bigger. We formally state such a hypothesis as follows.

\begin{hyp}
\label{h2}
Given a signed graph $G$ and a tolerance parameter $\beta$, suppose the optimal subgraph and coloring found by Algorithm~\ref{alg:hillclimbing} starting with vertex $a$ are $G_a$ and $\mathcal{X}_a$, starting with $b$ are $G_b$ and $\mathcal{X}_b$.
If $\hat{\Phi}(G_a, \beta, \mathcal{X}_a) \geq \hat{\Phi}(G_b, \beta, \mathcal{X}_b)$, we have $|V(G_b)| \leq 2|V(G_a)|$.
\end{hyp}

With these two hypotheses, we have the following lemma that describes a sampling strategy that is able to find a $(1-\epsilon)$-optimal subgraph within an acceptable number of calls of the search process.
The proof is deferred to Appendix~\ref{appendix:sec4}.

\begin{lemma}
\label{lemma:sample}
    Given a signed graph $G$, a tolerance parameter $\beta$, and a positive $\epsilon < 1$, suppose we run Algorithm~\ref{alg:hillclimbing} in $k$ iterations, where the $i$-th iteration starts with a uniformly sampled vertex $x_i \in V(G)$, and the optimal subgraph found is $G_{i}$.
    If \textsc{Hypothesis}~\ref{h1} and \textsc{Hypothesis}~\ref{h2} hold, the expected number of iterations that find a $(1-\epsilon)$-optimal subgraph $\expec{}{X}$ is $\Omega(\sum_{i=1}^{k} |V(G_i)| / |V(G)|)$.
\end{lemma}

\begin{algorithm}[t]
\caption{\textsc{RH: Sampling}}\label{alg:main}
\KwIn{Signed graph $G$; Tolerance parameter $\beta$; Iteration constant $\mathcal{C}$; Non-incremental probability $p$; Early stop turn limit $T$.}
\KwOut{$H \subseteq G$: the found balanced subgraph with $\beta$-tolerance.}
$H \leftarrow \emptyset$; \\ \label{sampling:inith}
$\hat{\Phi}_{opt} \leftarrow 0$; \\ \label{sampling:initphi}
$TotalSize \leftarrow 0$; \\ \label{sampling:inittotal}
\While {$TotalSize < \mathcal{C}|V(G)|$}{ \label{sampling:loop}
    $s \leftarrow$ Sample a vertex from $V(G)$; \\ \label{sampling:sample}
    $\overline{H}, \hat{\Phi} \leftarrow$ \texttt{Search}($G, \beta, s, p, T$); \\ \label{sampling:search}
    \If{$\hat{\Phi} > \Phi_{opt}$ }{ \label{sampling:update1}
          $\hat{\Phi} \leftarrow \hat{\Phi}_{opt}$; \\ \label{sampling:update2}
          $H \leftarrow \overline{H}$; \label{sampling:update3}
    }
    $TotalSize \leftarrow TotalSize + |V(\overline{H})|$; \label{sampling:updatetotalsize}
}
\Return $H$; \label{sampling:return}
\end{algorithm}

We provide an implementation of such sampling strategy in Algorithm~\ref{alg:main}, which is, in fact, an application of Lemma~\ref{lemma:sample}.
We use $H$, $\hat{\Phi}_{opt}$ to keep track of the current optimal subgraph and its corresponding $\hat{\Phi}$ value. They are initialized to be $\emptyset$ and $0$ in the beginning (line~\ref{sampling:inith} to ~\ref{sampling:initphi}). We also use a variable $TotalSize$ to keep track of the total size of all return subgraphs from Algorithm~\ref{alg:hillclimbing}, which is also initialized to be $0$ (line~\ref{sampling:inittotal}).

For each iteration, we randomly select a vertex $s$ (line~\ref{sampling:sample}) as the starting vertex and pass it into Algorithm~\ref{alg:hillclimbing} (line~\ref{sampling:search}). After receiving the result from Algorithm~\ref{alg:hillclimbing}, we update $H$, $\hat{\Phi}_{opt}$ (line~\ref{sampling:update2} to ~\ref{sampling:update3}) if the newly found subgraph is better (line~\ref{sampling:update1}). Before the new iteration, we accumulate the size of the newly found subgraph into $TotalSize$ (line~\ref{sampling:updatetotalsize}. 
The whole process will stop when $TotalSize \geq \mathcal{C}|V(G)|$ (line~\ref{sampling:loop}). By Lemma~\ref{lemma:sample}, we can set a proper termination condition by accumulating the subgraph size from each search process. More specifically, When $TotalSize$ reaches $\mathcal{C}|V(G)|$, it is expected to find a nearly optimal solution $\mathcal{C}$ times.
In the end, we return $H$ as the main result (line~\ref{sampling:return}).

\paragraph{Time Complexity Analysis} 
We show that the time complexity of Algorithm~\ref{alg:main} is $O(\Delta (\mathcal{C} + 1)n\log^2 n)$, where $n$ is the number of vertices in $G$.
We start with the time complexity of Algorithm~\ref{alg:hillclimbing}.

\begin{lemma}
Suppose that the loop starting on Line~\ref{search:loop} repeats for a total of $t$ times. Let $h$ be the size of $S$ on Line~\ref{search:return}. Then we have, with high probability, $h=\Omega(t/\log n)$ and that the running time of Algorithm~\ref{alg:hillclimbing} is $O(\Delta t \log n)$ in expectation.
\label{lemma:hillclimbingtime}
\end{lemma}
\begin{proof}
By Lemma~\ref{lemma:stop}, with high probability, $h=\Omega(t/\log n)$. 

Let $\Delta$ be the maximum degree of $G$. The running time of each iteration is dominated by the cost of updating the heaps and the time for \texttt{DelEval($S$)}. Updating the heaps cost $O(\Delta \log n)$ time. The function \texttt{DelEval($S$)} costs $O(\Delta|S|)$ time. It is called with probability $\frac{p\log |S|}{|S|}$ in each iteration. Thus, its expected runtime in each iteration is $O(\Delta \log n)$.
We have the running time of Algorithm~\ref{alg:hillclimbing} is $t\cdot O(\Delta \log n)=O(\Delta t\log n).$
\end{proof}

\begin{lemma}
With high probability, the expected time complexity of Algorithm~\ref{alg:main} is $O(\Delta (\mathcal{C} + 1)n\log^2 n)$.
\end{lemma}
\begin{proof}
Suppose we call Line~\ref{sampling:search} for $k$ times. For the $i$-th call, let $t_i$ be the number of iterations of the loop on Line~\ref{search:loop} and let $h_i$ ($1\le i\le k$) be the size of $S$ on Line~\ref{search:return}.
By Lemma~\ref{lemma:hillclimbingtime}, the total time complexity in expectation is $O(\sum_{i=1}^k \Delta t_i \log n)$, and the sum of $h_i$ satisfies $\sum_{i=1}^k h_i=\Omega(\sum_{i=1}^k t_i/\log n).$
Since $\sum_{i=1}^k h_i \leq (\mathcal{C} + 1)n$, we have that the total runtime of Algorithm~\ref{alg:main} in expectation is $O(\sum_{i=1}^k \Delta t_i \log n) = O(\Delta (\mathcal{C} + 1)n\log^2 n)$.
\end{proof}
\section{EVALUATION}
\label{sec:eval}
In this section, we address the following research questions to evaluate various important aspects of our algorithm:
\begin{itemize}[leftmargin=*]
    \item \textbf{RQ1 (Effectiveness - TMBS): } Given various $\beta$, what are the optimal subgraphs in terms of size, and \textsc{Tolerant Index Count}, found by our method and baselines?
    \item \textbf{RQ2 (Effectiveness - MBS): } What are our method's and baselines' performances in finding the maximum balanced subgraph?
    \item \textbf{RQ3 (Efficiency): } What are the runtimes for our method and the baselines, and how do they scale with very large networks?
    \item \textbf{RQ4 (Generalizability): } Can our model and method be adapted to other related tasks in the signed networks?
\end{itemize}
We also conduct three experiments shown in appendix~\ref{appendix:experiments}, which are designed to validate our hypotheses, determine optimal hyperparameters, and assess stability.   


\subsection{Experimental Setting}
\label{sec:ExperimentalSetting}

\paragraph{Baselines}
We compare our method with the baselines from highly related works (\textit{e.g.,} MBS and polarized community detection), including spectral and other heuristic methods.

Note that since finding tolerant balanced subgraphs is typically more challenging than previous community detection tasks in signed networks, we also adapt our TBC-relaxation in Section~\ref{sec:relaxation} for all baseline methods to ensure that they can return the subgraphs satisfying the tolerance constraint.
We summarize the core ideas of each baseline method as follows.

\begin{itemize}[leftmargin=*]
    \item \textsc{Eigen}: Based on a spectral method from~\cite{eigen}, we first compute $\mathbf{v}$, the eigenvector of the signed adjacency matrix $A$ corresponding to the largest eigen value $\lambda$. 
    For each vertex $i$, if a Bernoulli experiment with success probability $p = |\mathbf{v}_i|$ is successful, we assign vertex $i$ a color determined by $\text{sgn}(\mathbf{v}_i)$. 
    The maximum connected components with non-negative $\hat{\Phi}$ serve as solutions to \textsc{Problem}~\ref{p3} and \textsc{Problem}~\ref{p4}, while the connected component with the maximum $\hat{\Phi}$ provides the solution to \textsc{Problem}~\ref{p5}.
    
    \item \textsc{Timbal}: Based on a state-of-the-art method for MBS from~\cite{ordozgoiti2020finding}, we start from the entire graph and repeatedly remove vertices from it, guided by an eigenvalue approximation outlined in their paper.
    Specifically, we compute $\hat{\Phi}$ by the optimal coloring among two non-trivial methods~\cite{eigen,GULPINAR2004359} for TMBS problems.
    
    \item \textsc{GreSt}: This method combines a classical and effective algorithm in dense subgraph~\cite{greedydense} and a heuristic coloring method for signed networks by spanning tree~\cite{GULPINAR2004359}. Firstly, we determine the coloring of the entire graph by a random spanning tree. We then repeatedly remove vertices maximizing $\hat{\Phi}$ of the remaining graph.
    A min-heap is used to speed up the process of finding vertices. To acquire the solution to different tasks, we need to store the deletion operations. Then we can restore the graph by undoing deletions one by one. Such reversal allows us to keep track of the optimal connected component efficiently, instead of scanning the whole graph after each deletion. 
    
    \item \textsc{RH (LS)}: Instead of using region-based sampling, we execute our local search in Section~\ref{sec:localsearch} for all starting vertices and select the optimal solutions, to investigate the effectiveness of the regional-based sampling. This method may provide a better solution than \textsc{RH} but is not efficient.

\end{itemize}

\begin{table}[t!]
\caption{Signed network datasets used in experiments, including the number of vertices ($|V|$) and edges ($|E|$), the ratio of negative edges ($\rho^{-}$), and the ratio of non-zero elements ($\delta$).}
\resizebox{0.9\linewidth}{!}{
\begin{tabular}{|c|rrrr|} 
\hline
Dataset            & \multicolumn{1}{c}{\footnotesize $|V|$} & \multicolumn{1}{c}{\footnotesize $|E^{+}\cup E^{-}|$} & \multicolumn{1}{c}{\tiny$\rho^{-}=\frac{|E^{-}|}{|E^{+}\cup E^{-}|}$} & \multicolumn{1}{c|}{\tiny$\delta=\frac{2|E^{+}\cup E^{-}|}{|V|(|V|-1)}$}  \\ 
\hline
\textsc{Bitcoin}   & 5k                        & 21k                                     & 0.15                                                             & 1.2e-03                                                              \\
\textsc{Epinions}  & 131k                      & 711k                                    & 0.17                                                             & 8.2e-05                                                              \\
\textsc{Slashdot}  & 82k                       & 500k                                    & 0.23                                                             & 1.4e-04                                                              \\
\textsc{Twitter}   & 10k                       & 251k                                    & 0.05                                                             & 4.2e-03                                                              \\
\textsc{Conflict}  & 116k                      & 2M                                      & 0.62                                                             & 2.9e-04                                                              \\
\textsc{Elections} & 7k                        & 100k                                    & 0.22                                                             & 3.9e-03                                                              \\
\textsc{Politics}  & 138k                      & 715k                                    & 0.12                                                             & 7.4e-05                                                              \\ 
\hline
\textsc{Growth}    & 1.87M                            & 40M                                         & 0.50                                                                   &  2.3e-05                                                                    \\
\hline
\end{tabular}
}
\label{tab:dataset}
\end{table}

\paragraph{Datasets}
We select 7 publicly-available real-world signed networks\footnote{From \url{konect.cc} and \url{snap.stanford.edu}}, \textsc{Bitcoin}, \textsc{Epinions}, \textsc{Slashdot}, \textsc{Twitter}, \textsc{Conflict}, \textsc{Elections}, and \textsc{Politics}, which were widely used in previous related works~\cite{ordozgoiti2020finding,eigen, CohesivelyPolarized}.
In addition, to investigate our methods' scalability, we generate \textsc{Growth}, a very large signed network induced from the real temporal network \textsc{Wikipedia-growth}\footnote{http://konect.cc/networks/wikipedia-growth}.
Specifically, we select a threshold $\tau$, and give positive signs for the edges with time stamp $t(e) \geq \tau$ and negative signs for the edges with $t(e) < \tau$.
By using a proper $\tau$, the ratio of negative edges of the induced graph can be $0.5$. 
In this network, a balanced graph contains two communities, where the edges within each community are recently formed, while the crossing edges are relatively old.
Detailed information for each dataset can be found in Table~\ref{tab:dataset}.

All experiments are conducted on a Ubuntu 22.04 LTS workstation, equipped with a 12th Gen Intel(R) Core(TM) i9-12900HX. 
We set hyperparameters $T = 20, p = 0.8, \mathcal{C} = 1.5$ for all experiments.

\subsection{Finding Balanced Subgraphs with Tolerance}
\label{sec:beta-eval}
We aim to identify sizable and polarized subgraphs while taking the tolerance into account (\textbf{RQ1}).
Firstly, by running algorithms in Section~\ref{sec:ExperimentalSetting}, we keep track of the optimal subgraph that we have found to \textsc{Problem}~\ref{p3}, \textsc{Problem}~\ref{p4}, and \textsc{Problem}~\ref{p5} respectively.
Noticeably, we do not modify either \textsc{RH} or \textsc{RH (LS)} for \textsc{Problem}~\ref{p3} and \textsc{Problem}~\ref{p4}. That is to say, we directly use our result on solving \textsc{Problem}~\ref{p5} to compare with other algorithms which are tailored for either \textsc{Problem}~\ref{p3} or \textsc{Problem}~\ref{p4}, which demonstrates our proposed \textsc{Problem}~\ref{p5} is a good approximation for the other two problems.

The comparison results for these three problems are shown in Figure~\ref{fig:beta-result-V}, Figure~\ref{fig:beta-result-E}, and Figure~\ref{fig:beta-result-T}, respectively.
In our experiments, we consider $15$ different tolerance parameters ($\beta_i = 2^{-i/2}$ for all $2 \leq i \leq 16$), where $i = 2$ implies allowance for all connected subgraphs, while $i = 16$ means that the found subgraphs are almost strictly balanced.
Since \textsc{Problem}~\ref{p3} and \textsc{Problem}~\ref{p4} become trivial when $\beta \geq \frac{L(G)}{|E^{+} \cup E^{-}|}$, we shade the corresponding ranges\footnote{As computing the exact value of $L(G)$ is \textbf{NP}-hard, we cannot calculate the exact range of $\beta$ corresponding to the trivial cases. Instead, we calculate an upper bound $\overline{L}$ of $L(G)$ by a promising valid solution and shade only the subrange $\beta \ge \frac{\overline{L}}{|E^{+} \cup E^{-}|}$.}  in Figure~\ref{fig:beta-result-V} and Figure~\ref{fig:beta-result-E}. 
We omit the results for methods that cannot finish in 100,000 seconds. 

Observing the experiment results, our proposed method \textsc{RH} outperforms all other baselines significantly in all three problems. 
This also demonstrates that our proposed tolerance model is a more realistic model for the general case of balanced signed graph models. 
In addition, we can see that \textsc{RH} produces results that are close to \textsc{RH (LS)}. 
Therefore, we manage to empirically support our hypothesis in Section~\ref{sec:sampling} in real-world data.

\begin{figure*}[h!]
\vspace{-1mm}
  \centering
  \includegraphics[width=0.83\linewidth]{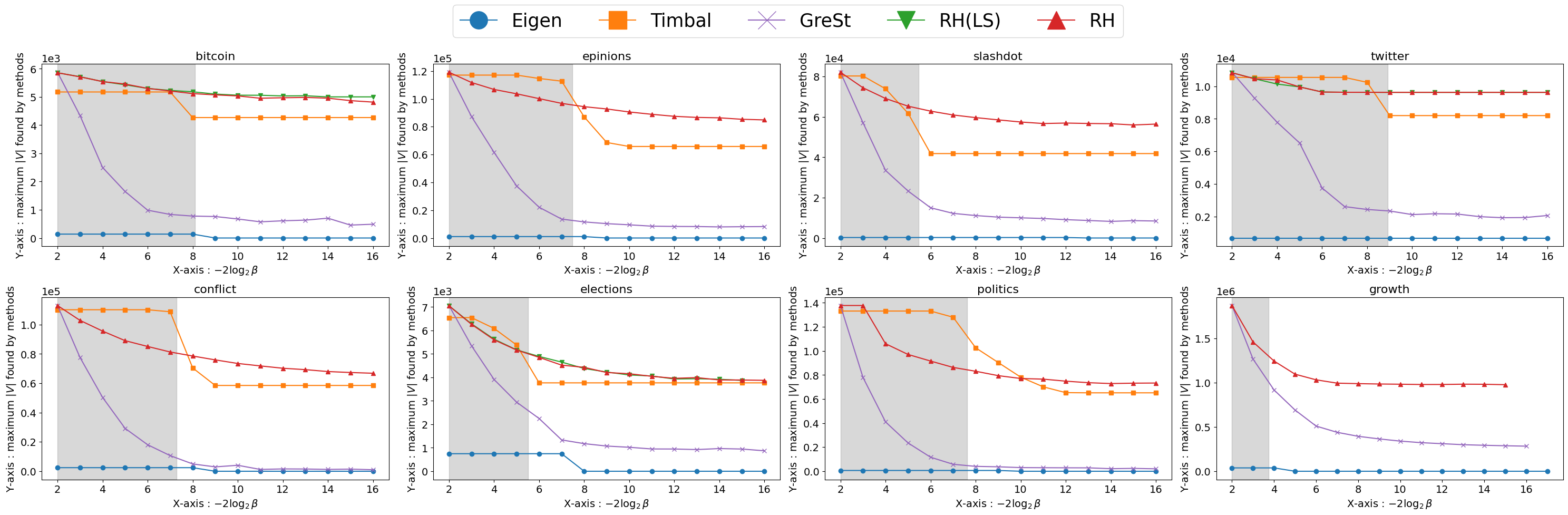}
  \caption{\textsc{Problem}~\ref{p3}: Comparing maximum tolerantly balanced subgraphs in vertex cardinality ($|V|$) across various tolerances ($\beta = 2^{-x/2}$),
  where the part corresponding to trivial $\beta$ values for this problem has been shaded.
  }
  \label{fig:beta-result-V}
\end{figure*}
\begin{figure*}[h!]
\vspace{-1mm}
  \centering
  \includegraphics[width=0.83\linewidth]{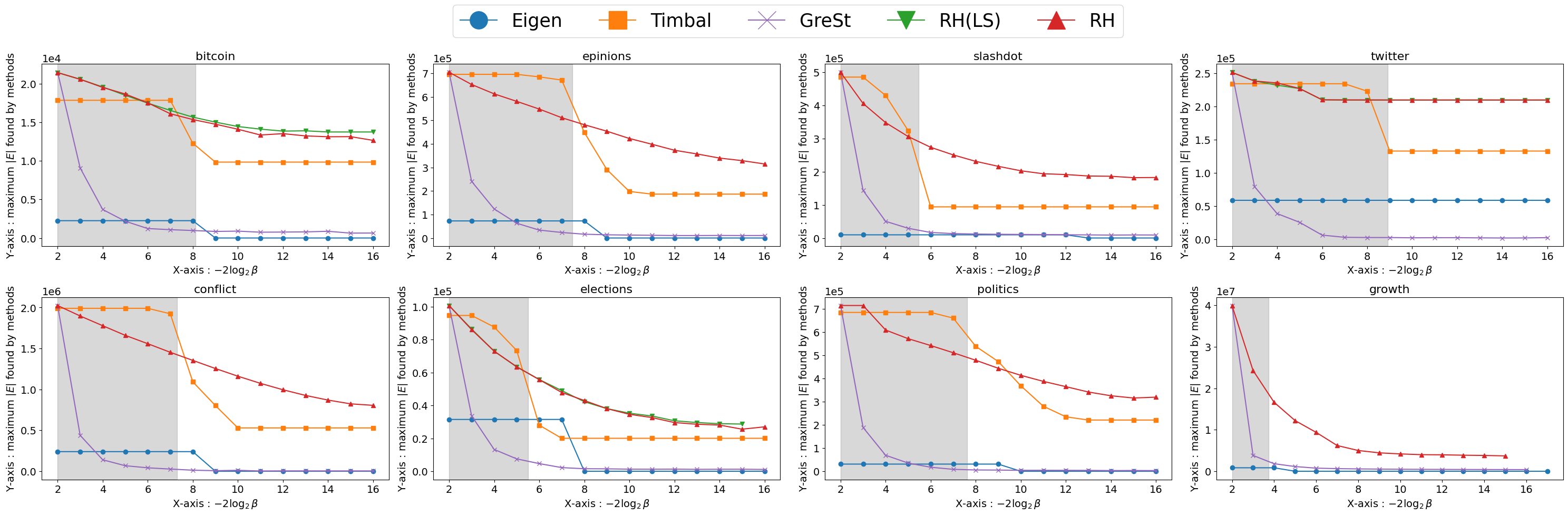}
  \caption{\textsc{Probelm}~\ref{p4}: Comparing maximum tolerantly balanced subgraphs in edge cardinality ($|E^{+} \cup E^{-}|$) across various tolerances ($\beta = 2^{-x/2}$),
  where the part corresponding to trivial $\beta$ values for this problem has been shaded.
  }
  \label{fig:beta-result-E}
\end{figure*}
\begin{figure*}[h!]
\vspace{-1mm}
  \centering
  \includegraphics[width=0.83\linewidth]{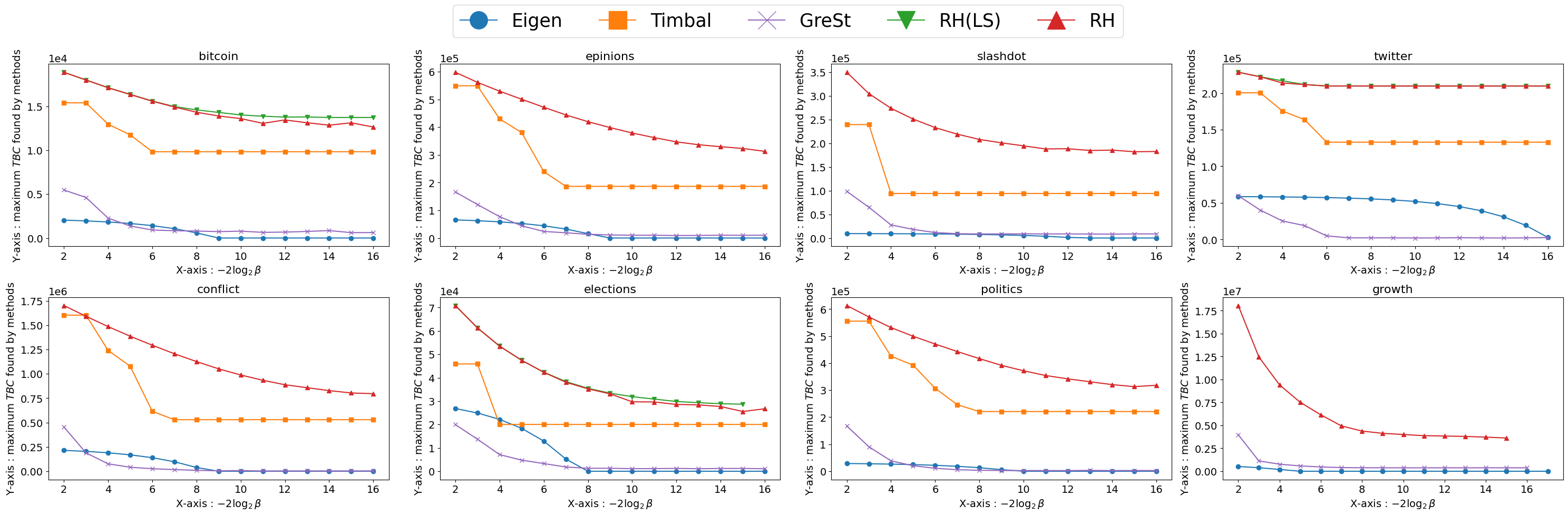}
  \caption{\textsc{Probelm}~\ref{p5}: Comparing maximum tolerantly balanced subgraphs in TBC ($\hat{\Phi}$) across various tolerances ($\beta = 2^{-x/2}$).}
  \label{fig:beta-result-T}
\end{figure*}
\subsection{Finding Strictly Balanced Subgraphs}
Finding the Maximum Balanced Subgraph (Problem~\ref{p1} and ~\ref{p2}) is an important and well-studied problem in signed networks~\cite{ordozgoiti2020finding} (\textbf{RQ2}).
As previously discussed, MBS problems are special forms of proposed TMBS problems (Problem~\ref{p3} and ~\ref{p4}) by setting $\beta < \frac{1}{|E^{+}\cup E^{-}|}$.

Table~\ref{tab:MBS} shows the results of our case study on the MBS problems, where we omit the results for \textsc{Eigen} since it cannot return subgraphs other than a single vertex. 
\textsc{RH} produces significantly better results than the previous state-of-the-art method \textsc{Timbal}, which is specifically designed for the MBS problems while ours is not.

\begin{table*}[t!]
\caption{Case study: Largest strictly balanced subgraph (in terms of $|V|$ or $|E|$) found by each method for each dataset ($\beta < \frac{1}{|E^{+}\cup E^{-}|}$), where {\em NA} denotes the corresponding method cannot finish in 100,000 seconds on the respective dataset.}
\centering
\resizebox{0.9\linewidth}{!}{
\label{tab:MBS}
\begin{tabular}{|c|rr|rr|rr|rr|rr|rr|rr|rr|} 
\hline
\multirow{2}{*}{Method} & \multicolumn{2}{c|}{\textsc{Bitcoin}}                  & \multicolumn{2}{c|}{\textsc{Epinions}}                 & \multicolumn{2}{c|}{\textsc{Slashdot}}                 & \multicolumn{2}{c|}{\textsc{Twitter}}                  & \multicolumn{2}{c|}{\textsc{Conflict}}                 & \multicolumn{2}{c|}{\textsc{Elections}}                & \multicolumn{2}{c|}{\textsc{Politics}}                 & \multicolumn{2}{c|}{\textsc{Growth}}                    \\
                        & \multicolumn{1}{c}{$|V|$} & \multicolumn{1}{c|}{$|E|$} & \multicolumn{1}{c}{$|V|$} & \multicolumn{1}{c|}{$|E|$} & \multicolumn{1}{c}{$|V|$} & \multicolumn{1}{c|}{$|E|$} & \multicolumn{1}{c}{$|V|$} & \multicolumn{1}{c|}{$|E|$} & \multicolumn{1}{c}{$|V|$} & \multicolumn{1}{c|}{$|E|$} & \multicolumn{1}{c}{$|V|$} & \multicolumn{1}{c|}{$|E|$} & \multicolumn{1}{c}{$|V|$} & \multicolumn{1}{c|}{$|E|$} & \multicolumn{1}{c}{$|V|$} & \multicolumn{1}{c|}{$|E|$}  \\ 
\hline
\textsc{Timbal}         & 4,050                     & 9,757                       & \underline{65,879}                    & \underline{190,751}                    & \underline{43,742}                    & \underline{101,590}                    & 8,636                     & 129,927                    & \underline{51,463}                    & \underline{361,663}                    & 3,579                     & 17,633                     & \underline{65,188}                    & \underline{230,529}                    & \textit{NA}               & \textit{NA}                 \\
\textsc{GreSt}          & 571                       & 713                        & 7,078                     & 9,229                      & 7,923                     & 8,465                      & 1,895                     & 2,277                      & 759                       & 4,885                      & 935                       & 1,232                      & 2,371                      & 3,155                      & \underline{272,621}                   & \underline{375,603}                     \\
\textsc{RH (LS)}        & \textbf{5,002}            & \textbf{13,746}            & \textit{NA}               & \textit{NA}                & \textit{NA}               & \textit{NA}                & \textbf{9,628}            & \textbf{209,633}           & \textit{NA}               & \textit{NA}                & \textbf{3,970}            & \textbf{28,453}            & \textit{NA}               & \textit{NA}                & \textit{NA}               & \textit{NA}                 \\
\textsc{RH}             & \underline{4,935}                     & \underline{13,050}                     & \textbf{84,165}           & \textbf{302,152}           & \textbf{55,968}           & \textbf{181,069}           & \textbf{9,628}            & \textbf{209,633}           & \textbf{65,468}           & \textbf{746,640}           & \underline{3,926}                     & \underline{26,478}                     & \textbf{72,431}           & \textbf{317,384}           & \textbf{999,504}          & \textbf{3,415,441}          \\
\hline
\end{tabular}
}
\end{table*}

\subsection{Running Time Analysis}
\label{sec:runningtime}

This experiment is designed to study the efficiency and scalability of our method (\textbf{RQ3}). As shown in Table~\ref{tab:time}, when the graph size is up to ($|V| = 1.87$M, $|E^{+}\cup E^{-}| = 40$M), \textsc{RH} can still produce reasonable results within 1,000 seconds. On the other hand, the other four compared methods either produce much worse results while taking longer time or even fail to finish execution in 100,000 seconds. In all, the empirical results demonstrate the efficiency of our method in real-world data in addition to its asymptotic theoretical complexity.
\subsection{Solving the 2PC Problem}
\label{sec:2pc}
This experiment aims to evaluate the generalizability when adapting RH and our tolerant model to different variants of polarity community detection (\textbf{RQ4}).
In addition to balance-related problems, \textsc{2-Polarized-Communities} (2PC), proposed by Bonchi and Galimberti~\cite{eigen}, serves as another model for polarized community detection.
In their model, the measurement of polarity is penalized by the size of the solution:

\vspace{-0.5ex}
\begin{problem}[2PC]
\label{p6}
Given a signed graph $G$ with signed adjacency matrix $A$, find a vector $x \in \{-1, 0, 1\}$ that maximizes $\frac{x^{T}Ax}{x^{T}x}$.
\end{problem}
\vspace{-0.5ex}

To solve the 2PC problem by our tolerance balance model, we add an additional penalty term to the solution size in the tolerant balance count:
We define $\hat{\Phi}'(G, \beta, \mathcal{X}, \rho) = \hat{\Phi}(G, \beta, \mathcal{X}) - \rho|V|$ to serve as a new object function in our \text{RH} algorithm. Throughout the algorithm, we set $\beta = \frac{1}{2}$. This is because when $\beta = \frac{1}{2}$, if $\hat{\Phi}'\geq 0$, the polarity is no less than $\rho$. Therefore, \textsc{Problem}~\ref{p6} can be solved by the new \textsc{RH} algorithm, where we apply an iterative mechanism on $\rho$. The algorithm details can be found in Appendix~\ref{appedix:2pc}.

We compare RH's results with the optimal result computed by the two methods in Bonchi and Galimberti's paper~\cite{eigen} (\textit{i.e.,} \textsc{Eigen} and \textsc{Greedy}). 
The results are shown in Table~\ref{tab:MBDS}.
RH's slight modifications efficiently yield promising communities, demonstrating the adaptability and effectiveness of our tolerant balance model and algorithm in varied polarity community detection scenarios.

\begin{table}[H]
\caption{Mean running time (s) for each method on the various $\beta$ in Section~\ref{sec:beta-eval}.}
\resizebox{0.9\linewidth}{!}{
\begin{tabular}{|c|rrrrr|} 
\hline
Dataset            & \multicolumn{1}{c}{\textsc{Eigen}} & \multicolumn{1}{c}{\textsc{Timbal}} & \multicolumn{1}{c}{\textsc{GreSt}} & \multicolumn{1}{c}{\textsc{RH (LS)}} & \multicolumn{1}{c|}{\textsc{RH}}  \\ 
\hline
\textsc{Bitcoin}   & 0.979                              & 2.670                               & 0.018                              & 135.040                              & 0.057                             \\
\textsc{Epinions}  & 9.213                              & 220.919                             & 0.319                              & $\geq 100,000$                       & 4.162                             \\
\textsc{Slashdot}  & 7.287                              & 226.805                             & 0.218                              & $\geq 100,000$                       & 3.175                             \\
\textsc{Twitter}   & 1.704                              & 7.306                               & 0.085                              & 2,516.671                             & 0.493                             \\
\textsc{Conflict}  & 12.899                             & 310.736                             & 0.781                              & $\geq 100,000$                       & 9.031                             \\
\textsc{Elections} & 1.022                              & 4.542                               & 0.039                              & 675.484                              & 0.288                             \\
\textsc{Politics}  & 10.003                             & 143.605                             & 0.304                              & $\geq 100,000$                       & 4.782                             \\
\textsc{Growth}    & 297.299                            & $\geq 100,000$                      & 48.915                             & $\geq 100,000$                       & 813.099                           \\
\hline
\end{tabular}
}
\label{tab:time}
\end{table}

\begin{table}[H]
\centering
\caption{Optimal subgraphs in the term of 2PC-polarity ($\frac{x^{T}Ax}{x^{T}x}$) found by each method, associated with running time.}
\label{tab:MBDS}
\resizebox{0.95\linewidth}{!}{
\begin{tabular}{|c|rr|rr|rr|} 
\hline
\multirow{2}{*}{Dataset} & \multicolumn{2}{c|}{\textsc{RH}}                             & \multicolumn{2}{c|}{\textsc{Eigen}}                          & \multicolumn{2}{c|}{\textsc{Greedy}}                          \\
                         & \multicolumn{1}{c}{polarity} & \multicolumn{1}{c|}{time (s)} & \multicolumn{1}{c}{polarity} & \multicolumn{1}{c|}{time (s)} & \multicolumn{1}{c}{polarity} & \multicolumn{1}{c|}{time (s)}  \\ 
\hline
\textsc{Bitcoin}         & \textbf{14.82}               & \textbf{0.035}                & 14.76                        & 0.706                         & 14.50                        & 0.999                          \\
\textsc{Epinions}        & \textbf{85.27}               & \textbf{0.873}                & 64.36                        & 10.271                        & 85.15                        & 428.625                        \\
\textsc{Slashdot}        & 41.21                        & \textbf{0.584}                & 39.85                        & 7.325                         & \textbf{41.36}               & 154.383                        \\
\textsc{Twitter}         & \textbf{87.21}               & \textbf{0.818}                & 87.04                        & 1.717                         & 86.97                        & 7.853                          \\
\textsc{Conflict}        & \textbf{94.27}               & 15.134                        & 87.83                        & \textbf{13.144}               & 63.99                        & 552.629                        \\
\textsc{Elections}       & 36.27                        & \textbf{0.395}                & 35.87                        & 0.950                         & \textbf{36.34}               & 2.223                          \\
\textsc{Politics}        & 44.95                        & \textbf{0.862}                & 44.22                        & 9.931                         & \textbf{45.01}               & 475.353                        \\
\hline
\end{tabular}
}
\end{table}
\section{CONCLUSIONS}
\label{sec:conc}
This paper presents a new and versatile model to identify polarized communities in signed graphs. 
The model accommodates inherent imbalances in polarized communities through a tolerance feature.
Additionally, we propose a region-based heuristic algorithm.
Through a wide variety of experiments on graphs of up to 40 M edges, it demonstrates effectiveness and efficiency beyond the state-of-the-art methods in addressing both traditional and generalized MBS problems.
We also adapt our model and algorithm to 2PC, a related but essentially different task, further verifying their generalizability.

\section{REPRODUCIBILITY}
Codes for our methods and for reproducing all the experimental results are available at GitHub\footnote{https://github.com/joyemang33/RH-TMBS}.

\begin{acks}
Chenhao Ma was partially supported by NSFC under Grant 62302421, Basic and Applied Basic Research Fund in Guangdong Province under Grant 2023A1515011280, and the Guangdong Provincial Key Laboratory of Big Data Computing, the Chinese University of Hong Kong, Shenzhen. We thank Yixiang Fang (The Chinese University of Hong Kong, Shenzhen), Qingyu Shi (Hailiang Foreign Language School), and Xinwen Zhang (The Chinese University of Hong Kong, Shenzhen) for valuable advice on this project. 
\end{acks}

\clearpage


\bibliographystyle{ACM-Reference-Format}
\bibliography{References}

\appendix
\newpage
\section{OMITTED PROOF AND ALGORITHM}
\subsection{Omitted Proof in Section~\ref{sec:preliminaries}}
\label{appendix:sec3}
\begin{proof}[Proof of lemma~\ref{lemma:hardness}]
We prove the lemma by a reduction from the calculation of the Frustration Index.
If we can decide whether $G$ is balanced under $\beta$-tolerance within the polynomial runtime, we will be able to decide whether $L(G)$ is greater than any guessing value $k$ by setting $\beta = \frac{k}{|E^{+} \cup E^{-}|}$.
Combined with a binary search, we can compute the Frustration Index of $G$ within the polynomial runtime, which has been proven as a \textbf{NP}-hard problem~\cite{agarwal2005log}.
\end{proof}

\begin{proof}[Proof of lemma~\ref{lemma:samequestion}]
We first show that the solutions to \textsc{Problem}~\ref{p2} is a $\frac{1}{\Delta}$-approximations for the \textsc{Problem}~\ref{p1}, and the relations between \textsc{Problem}~\ref{p3} and \textsc{Problem}~\ref{p4} can be proven similarly.

Given a signed graph $G$, let the optimal subgraph to \textsc{Problem}~\ref{p1} be $G_1$ with $n_1$ vertices and $m_1$ edges, and the optimal subgraph to \textsc{Problem}~\ref{p2} be $G_2$ with $n_2$ vertices and $m_2$ edges. 
Since $G_1$ is connected, we have 
\begin{align}
\label{m1n1}
    m_1 \geq n_1
\end{align}
Since $\Delta$ is the maximum degree of vertices in $G$, we have 
\begin{align}
\label{m2n2}
    n_2 \geq \frac{m_2}{\Delta}
\end{align}
Combined with (\ref{m1n1}) and (\ref{m2n2}), we have $n_2 \geq \frac{m_2}{\Delta} \geq \frac{m_1}{\Delta} \geq \frac{n_1}{\Delta}$, which implies $G_2$ is a $\frac{1}{\Delta}$-approximations for the \textsc{Problem}~\ref{p1}.
\end{proof}

\subsection{Omitted Lemma and Proof in Section~\ref{sec:algo}}
\label{appendix:sec4}

\begin{proof}[Proof of lemma~\ref{lemma:upperbound}]
The lemma is equivalent to that:

$\sum_{(i, j) \in E^{+}} \mathbbm{1}\{x_i \neq x_j\} + \sum_{(i, j) \in E^{-}} \mathbbm{1}\{x_i = x_j\}$ is a tight upper bound on $L(G)$.

First, given a coloring $\mathcal{X}$, if we remove all edges that let the indicator functions be true from $G$, the remaining graph will be balanced for each connected component due to \textsc{Definition}~\ref{def:balanced}.
This implies that $\sum_{(i, j) \in E^{+}} \mathbbm{1}\{x_i \neq x_j\} + \sum_{(i, j) \in E^{-}} \mathbbm{1}\{x_i = x_j\}$ is no less than the $L(G)$, which is the minimum number of edges to be removed.

We then show the bound is tight.
Consider the remaining graph $G'$ from removing all edges contributing to $L(G)$.
Since all connected components of $G'$ are balanced, we can find a partition $V = V_1 \cup V_2$ and $V_1 \neq V_2$ such that if we assign $x_i = 0$ for all $i \in V_1$ and $x_j = 1$ for all $j \in V_2$, the value of $\sum_{(i, j) \in E^{+}} \mathbbm{1}\{x_i \neq x_j\} + \sum_{(i, j) \in E^{-}} \mathbbm{1}\{x_i = x_j\}$ will become to zero for $G'$.
This implies the same value for $G$ under the coloring is equal to $L(G)$ and thus the bound is tight.
\end{proof}

\begin{proof}[Proof of lemma~\ref{lemma:sample}]
Suppose the optimal subgraph found by Algorithm~\ref{alg:hillclimbing} among all starting vertex is $G_{opt}$ and the probability of Algorithm~\ref{alg:main} finding a $(1-\epsilon)$-optimal subgraph in one iteration is $\Pr[\text{Sucess}]$.
By \textsc{Hypothesis}~\ref{h1}, we have \textit{(i)} $\Pr[\text{Sucess}] \geq \frac{|V(G_{opt})|}{2|V(G)|}$. 
Let $\mathcal{C} = \frac{\sum_{i=1}^{k} |V(G_i)|}{|V(G)|}$, and then we have $\sum_{i=1}^{k} |V(G_{i})| = \mathcal{C}|V(G)|$.
By \textsc{Hypothesis}~\ref{h2}, $\sum_{i=1}^{k} 2|V(G_{opt})| \geq \mathcal{C}|V(G)|$, and thus  \textit{(ii)} $k \geq \frac{\mathcal{C}}{2}\frac{|V(G)|}{|V{(G_{opt})}|}$.
Combining \textit{(i)} and \textit{(ii)}, we have $\expec{}{X} = k\Pr[\text{Sucess}]\geq \frac{\mathcal{C}}{2}\frac{|V(G)|}{|V{(G_{opt})}|} \cdot \frac{|V(G_{opt})|}{2|V(G)|} = \frac{\mathcal{C}}{4} = \Omega(\mathcal{C})$.
\end{proof}

\begin{lemma}
With high probability, for all valid $i$, after the $i$-th iteration of the on Line~\ref{search:loop} of Algorithm~\ref{alg:hillclimbing}, $|S|=\Omega(\frac{i}{\log n})$. 

\label{lemma:stop}
\end{lemma}
\begin{proof}

    Let $x$ denote the value of $|S|$. The initial value of $x$ is $1$. $x$ changes according to the following rules in each iteration:
    \begin{itemize}[leftmargin=*]
        \item With probability $(1-p)(1-p\frac{\log x}{x})$, $x$ increases by $1$;
        \item With probability $p\frac{\log x}{x}$, x increases by $1$, $0$, or $-1$;
        \item With probability $p(1-p\frac{\log x}{x})$, x increases by $1$, or $0$;
        \item The process ends when $x$ reaches $n$.
    \end{itemize}
    Suppose the process lasts for $t$ steps. We shall prove that with high probability, it holds simultaneously for all $i$ from $1$ to $t$ that the value of $x$ after the first $i$ steps is $\Omega(\frac{i}{\log n})$. 

    We denote $(1-p)(1-p\frac{\log x}{x})$ by $p^+(x)$, $p(1-p\frac{\log x}{x})$ by $p^0(x)$, and $p\frac{\log x}{x}$ by $p^-(x)$. 
    
    Let $next(x)$ be the value of $x$ in the next step.

    Intuitively, $x$ increases over steps because $p^+(x)$ is much larger than $p^-(x)$ for large enough $x$. Thus, we will partition steps into two parts by the value of $x$. Let $x_0$ be the smallest $x$ such that $p^-(x) < p^+(x)/2$. When $x$ is at least $x_0$, $x$ increases in expectation. When $x$ is smaller than $x_0$, $x$ will reach $x_0$ soon.
    
    We can solve for such $x_0$ by the definition of $p^+(x)$ and $p^-(x)$. We have that $x_0$ is a constant depending on $p$ satisfying $\frac{\log x_0}{x_0}=\frac{1-p}{3p-p^2}.$ Thus, for any $x\ge x_0$, $\expec{}{next(x)}\ge x+\frac{1-p}{3-p}.$

    We first sketch the proof. Consider the first $i$ iterations. We partition the $i$ iterations into $\Theta(\log(n))$ intervals with $i/\log n$ length each. At the beginning of each interval, with high probability, within $O(\log n)$ iterations, $x$ reaches a value that is at least $x_0.$ After reaching $x_0$, $x$ will never go below $x_0$ during the current interval with a constant probability. Conditioning on $x$ never going below $x_0$, the expected increment of $x$ in the current interval is at least $c l$ for some constant $c$ where $l$ is the number of remaining iterations in the interval. $l$ is at least $k/\log n-O(\log n).$ Because $x$ may increase by at most $l$, we have, by Markov bound, $\pr{}{\text{$x$ increases by no more than $cl/10$}}\le \frac{10-10c}{10-c}.$ In other words, with probability at least $1-\frac{10-10c}{10-c},$ the value of $x$ at the end of the current interval is at least $cl/10.$
    
    In summary, for each interval, the event that the value of $x$ at the end of the interval is at least $cl/10$ happens with constant probability where $l$ is at least $\Omega(i/\log n).$ If the event does not happen, we move on to the beginning of the next interval. Since we have $\Omega(\log n)$ intervals, with high probability, the event happens at least once. Once the event happens, we use Hoeffding's inequality to guarantee that the value of $x$ is $\Omega(l)$ after iteration $i$. 

    Now we prove each part of the sketch in detail. 
    \begin{itemize}[leftmargin=*]
        \item At the beginning of each interval, with high probability, within $O(\log n)$ iterations, $x$ reaches a value that is at least $x_0:$ Let $c$ be the maximum expected number of iterations for the variable $x$ to reach $x_0$ from any initial value $x_{init}<x_0$. $c$ is a constant depending on $p$. By Markov bound, we have that in every $2c$ iterations, with probability at least $1/2$, there exists at least one iteration before which $x$ is at least $x_0$. Thus, in $O(\log n)$ iterations, with high probability, $x$ reaches some value at least $x_0$ at least once.
        \item After reaching $x_0$, $x$ will never go below $x_0$ in the current interval with a constant probability: We will prove a stronger proposition that with constant probability, $x$ never goes below $x_0$ after arbitrary number of iterations. For this, we consider another random process where a random variable $y$ increases with probability $q\defeq \frac{1}{2}+\frac{1-p}{2(3-p)}$ and decreases with probability $1-q=\frac{1}{2}-\frac{1-p}{2(3-p)}$. If $x$ and $y$ starts with the same value, the probability that $x$ never goes below $x_0$ is lower bounded by the probability that $y$ never goes below $x_0$, because for each value $v\ge x_0,$ if both $x$ and $y$ start at $v$, $\pr{}{\text{$x$ reaches $v+1$ before $v-1$}}\ge \pr{}{\text{$y$ reaches $v+1$ before $v-1$}}.$ (Equality holds when $p^+(v)=q$ and $p^-(v)=1-q.$) Now we may consider the random process of $y.$ Let $p_{i,j}$ ($i\in \{-1,0\}, j\in \{0,1\}$) denote the probability that $y$ eventually reaches $x_0+i$ starting with $y=x_0+j$. Then we have $p_{-1, 0}= q+(1-q)p_{-1, 1}.$ We also have $p_{-1, 1}=p_{-1, 0}p_{0, 1}=p_{-1, 0}^2.$ These solve to $p_{-1, 0}=\frac{1-\sqrt{1-4q(1-q)}}{2-2q}$ or $p_{-1, 0}=1.$ The later case is impossible for the following reason: Let $p_{i, j}(d)$ be the probability that  $y$ eventually reaches $x_0+i$ starting with $y=x_0+j$ with in $d$ iterations. Denote $\frac{1-\sqrt{1-4q(1-q)}}{2-2q}$ by $r.$ $p_{-1, 0}(1)=q \le r.$ If $p_{-1, 0}(d)\le r, $ $p_{-1, 0}(d+1)=q+(1-q)p_{-1, 1}(d)\le q+(1-q)p_{-1, 0}(d)p_{0, 1}(d)=q+(1-q)(p_{-1,0}(d))^2 \le r$ where the last inequality is by the fact that $r$ is a root of $q+(1-q)r^2=r.$

        $1-r$ lower bounds the probability that $x$ never goes below $x_0$ after arbitrary number of iterations.
        \item Conditioning on $x$ never going below $x_0$, the expected increment of $x$ in the current interval is at least $c l$ for some constant $c$ where $l$ is the number of iterations remaining in the current interval: We know that for any $x\ge x_0$, $\expec{}{next(x)}\ge x+\frac{1-p}{3-p}.$ Setting $c=\frac{1-p}{3-p}$ completes the proof.
        \item Once the value of $x$ is at least $cl/10$ at the end of some interval, $x$ will be $\Omega(i/\log n)$ after iteration $i$: Consider the following $h=cl/20$ steps after the interval at the end of which $x$ is at least $cl/10.$ At each of the $h$ steps, $x$ increases by at least $c$ in expectation because $x\ge cl/20$ during these steps. Let $x_{\texttt{before}}$ ($x_{\texttt{after}}$) be the value of $x$ before (after) the $h$ steps. We have $\expec{}{x_{\texttt{after}}}=x_{\texttt{before}}+c^2l/20.$ We use Hoeffding's inequality to lower bound $x_{\texttt{after}}.$ 
        \begin{align*}
        &\pr{}{x_{\texttt{after}} < cl/10}\\
        \le &\pr{}{|x_{\texttt{after}}-\expec{}{x_{\texttt{after}}}| > \expec{}{x_{\texttt{after}}}-cl/10}\\
        \le &\pr{}{|x_{\texttt{after}}-\expec{}{x_{\texttt{after}}}| > c^2l/20}\\
        \le & \exp(-\Omega(l)).
        \end{align*} 
        In other words, with high probability, after the $h$ steps, the value of $x$ is still at least $cl/10.$ We can repeat the argument for $i/h=O(\log n)$ times to conclude that after the first $i$ iterations, the value of $x$ is at least $cl/10.$
    \end{itemize}





\end{proof}

\subsection{Omitted Algorithm in Section~\ref{sec:2pc}}
\label{appedix:2pc}
\begin{algorithm}[h!]
\caption{\textsc{RH: Iterative mechanism for 2PC}}\label{alg:MDBS}
\KwIn{Signed graph $G$;}
\KwOut{$H \subseteq G$: the found subgraph with high 2PC polarity; $\rho:$ the 2PC polarity of $H$;}
$H \leftarrow \emptyset$; \label{2pc:inith}
$\rho \leftarrow 0$; \\ \label{2pc:initrho}
$s \leftarrow $ Sample a vertex from $V(G)$;\\ \label{2pc:firssample}
\While {true}{ \label{2pc:loop}
$\overline{H}, \hat{\rho} \leftarrow $ \texttt{Search}($G, \beta = \frac{1}{2}, \sigma = 0.9\rho, s$)\tcp*{guided by the new object function $\hat{\Phi}' = \hat{\Phi} - \sigma|S| $. hyperparameters $T$ and $p$ are omitted.} \label{2pc:search}
\If{$\hat{\rho} > \rho$}{ \label{2pc:update1}
$\rho \leftarrow \hat{\rho}$;\\ \label{2pc:update2}
$H \leftarrow \overline{H}$; \label{2pc:update3}
}
\lElse{
$s \leftarrow $ Sample a vertex from $V(G)$ \label{2pc:newsample}
}
\lIf{\text{meets stop-conditions}}{break} \label{2pc:stop}
remove all vertice with degree less than $\rho$ from $G$; \label{2pc:optimize}
}
\Return $H, \rho$; \label{2pc:return}
\end{algorithm}

To compute solutions for the 2PC problem (Problem~\ref{p6}), we design an iterative mechanism (Algorithm~\ref{alg:MDBS}) for calling the search problem (Algorithm~\ref{alg:hillclimbing}) properly, replacing Algorithm~\ref{alg:main}. 

$H$ is for storing the current optimal subgraph, initialized as $\emptyset$ (line~\ref{2pc:inith}). Our iterative strategy is related with the current optimal polarity value $\rho$. We initialize it as $0$ (line~\ref{2pc:initrho}) in the beginning.

We pick a starting vertex $s$ at first, sampling from all vertices in $V(G)$ (line~\ref{2pc:firssample}). 
Each turn, we execute a search process with specific hyperparameters to try to find a new candidate answer (line~\ref{2pc:search}). 
We chose the tolerance parameter $\beta$ to be $\frac{1}{2}$. 
Since our new modified object function is $\hat{\Phi}' = \hat{\Phi} - \sigma|S| $, we also need to pass in $\sigma$ as $0.9\rho$. 
The search process will have to return the corresponding polarity $\hat{\rho}$ along with the found subgraph. We compare $\hat{\rho}$ with the previous found $\rho$. If it is better, we update $\rho$ and $H$ correspondingly (line~\ref{2pc:update1} to ~\ref{2pc:update3}). Note that we do not sample a new starting vertex each time. Instead, we also do so when the search process does not return a better result (line~\ref{2pc:newsample}). We regard this as a sign showing that we have reached the limit of the current chosen starting vertex.

The whole process will stop when a certain condition is met (line~\ref{2pc:stop}). The detail is omitted and can be referred to in our provided implementation. There are a few optimizations of such iteration mechanisms that can help improve the performance. Note that if the current polarity is $\rho$, for any vertex with a degree less than $\rho$, removing them will not worsen the answer. Therefore, we remove all such vertices from $G$ after the $\rho$ gets updated (line~\ref{2pc:optimize}).

\section{ADDITIONAL EXPERIMENTS}
\label{appendix:experiments}
\setcounter{table}{0}
\renewcommand{\thetable}{B.\arabic{table}}

\setcounter{figure}{0}
\renewcommand{\thefigure}{B.\arabic{figure}}

\subsection{Stability Study}
In our search process (Algorithm~\ref{alg:hillclimbing}), we will execute either vertex deletion or color flipping in a certain probability $p$. 
We conduct some experiments to study how such non-incremental operations help improve the stability of our algorithm, shown in Table~\ref{tab:variance}.

We modify the algorithm \textsc{RH} to a version with only insertion operations, denoted as \textsc{RH (insertion only)}. On all datasets, we execute both algorithms for 100 rounds. As common, we compute the variance $\sigma^2$ to measure the stability. 
Here, for each dataset, we also compute the maximum ($\max \hat{\Phi}$), minimum ($\min \hat{\Phi}$), and average ($\mu$) values for tolerant balance count (TBC). 

As shown in Table~\ref{tab:variance}, we can see with these two operations, the results are more stable and consistent, which can be observed in the rightmost column.

\begin{figure*}[h!]
\vspace{-1mm}
\centering
  \includegraphics[width=1\linewidth]{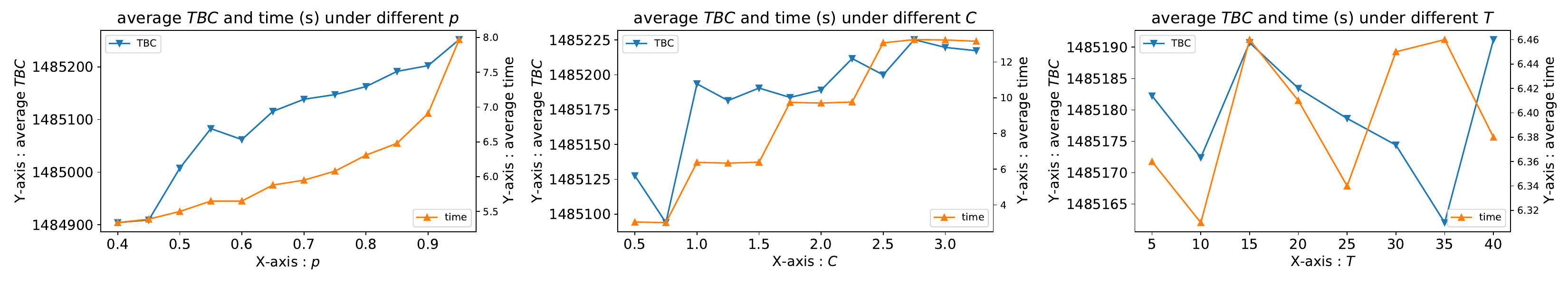}
  \caption{We study the influence of hyperparameters (\textit{i.e.,} Non-incremental probability $p$, Iteration constant $\mathcal{C}$ and Early stop turn limit $T$) on \textsc{Confilict} dataset with $\beta = 1/4$: For each hyperparameter, we execute the algorithm when setting the hyperparameter as various values while keeping other hyperparameters fixed. The results include the average Tolerant Balance Count (TBC) and running time in 25 rounds.}
  \label{fig:hyperparameters}

\end{figure*}

\begin{table*}[h!]
\centering
\caption{Stability study of Tolerant Balance Count (\textsc{TBC}) under $\beta = 1/8$ as example: We conduct 100 rounds of \textsc{RH} and \textsc{RH (insertion only)} on each dataset. The results include maximum TBC ($\max \hat{\Phi}$), minimum TBC ($\min \hat{\Phi}$), mean TBC ($\mu$), and the variance ($\sigma^2$) across the 100 rounds. Additionally, we present the variance reduction achieved by considering \textit{non-incremental operations}. }
\label{tab:variance}
\resizebox{0.75\linewidth}{!}{
\begin{tabular}{|c|rrrr|rrrr|r|} 
\hline
\multirow{2}{*}{Dataset} & \multicolumn{4}{c|}{\textsc{\textbf{RH (insertion only)}}}                                                                                           & \multicolumn{4}{c|}{\textsc{\textbf{RH}}}                                                                                                            & \multicolumn{1}{l|}{\multirow{2}{*}{\makecell[c]{\textbf{Variance} \\ \textbf{reduction}}}}  \\
                         & \multicolumn{1}{c}{$\min \hat{\Phi}$} & \multicolumn{1}{c}{$\max \hat{\Phi}$} & \multicolumn{1}{c}{$\mu$} & \multicolumn{1}{c|}{$\sigma^2$} & \multicolumn{1}{c}{$\min \hat{\Phi}$} & \multicolumn{1}{c}{$\max \hat{\Phi}$} & \multicolumn{1}{c}{$\mu$} & \multicolumn{1}{c|}{$\sigma^2$} & \multicolumn{1}{l|}{}                                        \\ 
\hline
\textsc{Bitcoin}         & 14,864                                & 15,523                                & 15,499                    & 4,977                           & 15,578                                & 15,619                                & 15,604                    & 25                              & \textbf{198$\times$}                                                \\
\textsc{Epinions}        & 462,237                               & 471,120                               & 470,536                   & 1,301,027                       & 472,158                               & 472,240                               & 472,195                   & 315                             & \textbf{4,129$\times$}                                               \\
\textsc{Slashdot}        & 223,054                               & 232,757                               & 232,110                   & 1,564,480                       & 231,988                               & 233,149                               & 232,957                   & 17,521                          & \textbf{88$\times$}                                                 \\
\textsc{Twitter}         & 187,549                               & 209,701                               & 208,051                   & 30,265,009                      & 197,072                               & 209,701                               & 209,097                   & 4,455,495                       & \textbf{5.7$\times$}                                                \\
\textsc{Conflict}        & 1,284,391                             & 1,289,332                             & 1,288,405                 & 430,597                         & 1,292,253                             & 1,293,741                             & 1,293,318                 & 84,706                          & \textbf{4$\times$}                                                  \\
\textsc{Elections}       & 40,819                                & 42,187                                & 41,981                    & 33,946                          & 42,206                                & 42,358                                & 42,276                    & 1,031                           & \textbf{29$\times$}                                                 \\
\textsc{Politics}        & 470,347                               & 470,764                               & 470,576                   & 4,289                           & 470,534                               & 470,905                               & 470,642                   & 4,042                           & \textbf{0.06$\times$}                                               \\
\hline
\end{tabular}
}
\end{table*}

\subsection{Statistical Hypothesis Testing}

In our sampling strategy and performance analysis, we propose two hypotheses:~\ref {h1} and ~\ref{h2}. Here, we conduct some experiments on 3 datasets and 4 different $\beta$ to verify if these hypotheses are reasonable.

For \textsc{Hypothesis}~\ref{h1}, we compute the maximum $1-\epsilon$ such that there exists a subset $V' \subseteq V(G_{opt})$ with $\frac{|V'|}{|V(G_{opt})|} \geq \frac{1}{2}$ such that
$\Phi(G_{x}, \beta) \geq (1-\epsilon) \Phi(G_{opt}, \beta), \forall x \in V'$. 
The result is shown in Table~\ref{tab:hypothesis1}.
The hypothesis is supported by the fact that all values are very close to $1$ and not sensitive to the value of $\beta$, which indicates $\epsilon$'s lower bound approaches to $0$.

To verify \textsc{Hypothesis}~\ref{h2}, we run Algorithm~\ref{alg:hillclimbing} from every possible starting vertex $s$ (similar to RH (LS) in Section~\ref{sec:eval}). For each setting, we compute $\max \frac{|V(G_b)|}{|V(G_a)|}$ for any two starting vertices $a,b$ satisfying $\hat{\Phi}(G_a, \beta, \mathcal{X}_a) \geq \hat{\Phi}(G_b, \beta, \mathcal{X}_b)$. 
The result is shown in Table~\ref{tab:hypothesis2}, where all values are very close to $1$ and far from the bound of $2$ in the hypothesis.
Therefore, Hypothesis~\ref{h2} is also empirically supported.

\begin{table}[t]
\centering
\caption{Verifing \textsc{Hypothesis}~\ref{h1}: For each dataset, we compute the maximum $1-\epsilon$ such that there exists a subset $V' \subseteq V(G_{opt})$ with $\frac{|V'|}{|V(G_{opt})|} \geq \frac{1}{2}$ such that $\Phi(G_{x}, \beta) \geq (1-\epsilon) \Phi(G_{opt}, \beta), \forall x \in V'$ under various $\beta$. }
\label{tab:hypothesis1}
\resizebox{0.7\linewidth}{!}{
\begin{tabular}{|c|rrrr|} 
\hline
Dataset   & $\beta=\frac{1}{4}$ & $\beta=\frac{1}{8}$ & $\beta=\frac{1}{16}$ & $\beta=\frac{1}{32}$ \\
\hline
\textsc{Bitcoin}   & 1.000               & 0.999               & 0.982                & 0.948                \\
\textsc{Twitter}   & 0.987               & 1.000               & 1.000                & 1.000                \\
\textsc{Elections} & 0.999               & 0.997               & 0.987                & 0.946                \\
\hline
\end{tabular}
}
\end{table}

\begin{table}[t]
\centering
\caption{Verifing \textsc{Hypothesis}~\ref{h2}: For each dataset, we compute $\max \frac{|V(G_b)|}{|V(G_a)|}$ for any two starting vertices $a,b$ satisfying $\hat{\Phi}(G_a, \beta, \mathcal{X}_a) \geq \hat{\Phi}(G_b, \beta, \mathcal{X}_b)$ under various $\beta$.}
\label{tab:hypothesis2}
\resizebox{0.7\linewidth}{!}{
\begin{tabular}{|c|rrrr|} 
\hline
Dataset   & $\beta=\frac{1}{4}$ & $\beta=\frac{1}{8}$ & $\beta=\frac{1}{16}$ & $\beta=\frac{1}{32}$ \\
\hline
\textsc{Bitcoin}   & 1.018               & 1.014               & 1.020                & 1.037                \\
\textsc{Twitter}   & 1.053               & 1.111               & 1.094                & 1.084                \\
\textsc{Elections} & 1.009               & 1.032               & 1.040                & 1.061                \\
\hline
\end{tabular}
}
\end{table}


\subsection{Hyperparameter Analysis}
For the three hyperparameters $p,C,T$ in Algorithm~\ref{alg:main}, we conduct a series of experiments to demonstrate that our choice is natural, shown in Figure~\ref{fig:hyperparameters}.

For each hyperparameter, we fix the other two unchanged as our option in Section~\ref{sec:algo} and enumerate its value in a range. For each value, we compute the average runtime and TBC over 25 times the executions of solving on the  \textsc{Confilict} with setting $\beta = \frac{1}{4}$. A larger $p$ may lead to more non-incremental operations, thus a longer runtime is expected. Similarly, a larger $C$ leads to more executions of the search process. Therefore, as $p,C$ gets larger, the runtime and the average TBC might be getting larger as well. This can be observed from the experiment, shown in Figure~\ref{fig:hyperparameters}. 

On the other hand, although such patterns can be observed, the difference between average runtime and TBC is not significant for all three hyperparameters. This indicates that our algorithm is robust and does not tune towards the dataset.

\end{document}